\documentclass[10pt, conference]{new-aiaa}
\usepackage{amsmath}

\usepackage{amssymb,amsfonts}
\usepackage{algorithmic}
\usepackage{graphicx}
\usepackage{blindtext}
\usepackage{hyperref}
\usepackage[english]{babel}

\usepackage{amsthm}
\usepackage{todonotes}
\setuptodonotes{inline}
\usepackage{xcolor}
\usepackage{subcaption}
\usepackage{longtable}

\newtheorem{theorem}{Theorem}[section]
\newtheorem{corollary}[theorem]{Corollary}
\newtheorem{lemma}[theorem]{Lemma}
\newtheorem{definition}[theorem]{Definition}
\newtheorem{proposition}[theorem]{Proposition}
\newtheorem{remark}[theorem]{Remark}
\newtheorem{assumption}{Assumption}

\usepackage{url}

\input{macros.sty}

\newcommand\blfootnote[1]{%
  \begingroup
  \renewcommand\thefootnote{}\footnote{#1}%
  \addtocounter{footnote}{-1}%
  \endgroup
}

\DeclareMathOperator{\dist}{dist}

\usepackage{textcomp}
\usepackage{xcolor}
\def\BibTeX{{\rm B\kern-.05em{\sc i\kern-.025em b}\kern-.08em
    T\kern-.1667em\lower.7ex\hbox{E}\kern-.125emX}}

\title{A Temporal Barrier Framework for Collision Avoidance in Multi-Agent Autonomous Aerial Vehicles
\blfootnote{DISTRIBUTION STATEMENT A. Approved for public release. Distribution is unlimited.}
}

\author{Benedikt Barthel Sorensen\footnote{Department of Mechanical Engineering, Massachusetts Institute of Technology, Cambridge, MA, USA, bbarthel@mit.edu}}
\affil{Massachusetts Institute of Technology, Cambridge, MA, 02139, USA}

\author{Mitchell Black and Erfaun Noorani\footnote{MIT Lincoln Laboratory, Lexington, MA, USA}}
\affil{MIT Lincoln Laboratory, Lexington, MA, 02421, USA}

\author{Themistoklis P. Sapsis\footnote{Department of Mechanical Engineering, Massachusetts Institute of Technology, Cambridge, MA, USA}}
\affil{Massachusetts Institute of Technology, Cambridge, MA, 02139, USA}

\begin{document}



\maketitle








\begin{abstract}
Operating teams of autonomous aircraft in dynamic, uncertain, and potentially
adversarial environments requires safety protocols that are reliable yet
selective, and allow agents to fly in close proximity while making progress
toward mission objectives.  We introduce adversarial time-to-collision
(aTTC), a risk metric that quantifies, for a given agent, how quickly
any surrounding agent could reach it assuming adversarial intent.
We embed aTTC into the control barrier function (CBF) framework, defining
the barrier directly in \textit{time} rather than distance or velocity.
The resulting aTTC-CBF is inherently \textit{anticipatory}: agents modulate
their own velocity based not on whether a peer is on a collision course,
but on how quickly one could reach collision given its dynamical
constraints. A differentiable neural-network surrogate makes the aTTC
computable in real time within a standard CBF quadratic program. 
Across long time-horizon simulations of 3D independent-pursuit and
formation-flight scenarios, the aTTC-CBF achieves up to twice the
waypoint progress at half the collision rate of a
higher-order distance-based CBF baseline. 
\end{abstract}

\section*{Nomenclature}
{\renewcommand\arraystretch{1.0}
\setlength{\LTleft}{0pt}
\setlength{\LTright}{\fill}
\noindent\begin{longtable*}{@{}l @{\quad=\quad} l@{}}
$a_{max}$        & maximum acceleration, km/s$^2$ \\
$b$              & neural-network layer bias vector \\
$d_{ij}$         & distance between agents $i$ and $j$, km \\
$F,\,G$          & joint (stacked) drift field and input matrix \\
$f_i,\,g_i$      & drift field and input matrix of agent $i$ \\
$G_\theta$       & neural-network surrogate model for the aTTC \\
$h$              & control barrier function, s \\
$L_F h,\,L_G h$  & Lie derivatives of $h$ along $F$ and $G$ \\
$\mathcal{L},\,\mathcal{L}_H$ & weighted training loss and Huber loss \\
$M$              & control-effort weighting matrix \\
$m,\,n,\,q$      & control, state, and disturbance dimensions \\
$N_a,\,N_e,\,N_p$& number of total, evader, and pursuer agents \\
$p$              & spatial position of an agent, km \\
$R$              & mission sphere radius, km \\
$r_a$            & CBF activation radius, km \\
$r_c$            & critical safety radius, km \\
$r_{col}$        & collision radius, km \\
$T$              & look-ahead horizon (error threshold), s \\
$T_{sim}$        & total simulation time, s \\
$t_{nm}$         & near-miss onset time, s \\
$\mathbf{u},\,u_i$ & joint and per-agent control input \\
$\mathbf{u}_0$   & nominal control input \\
$v$              & agent speed, km/s \\
$v_i,\,v_j$      & velocity of agents $i$ and $j$, km/s \\
$v_{max}$        & maximum (pursuer) speed, km/s \\
$v_{nom}$        & nominal cruise speed, km/s \\
$W$              & neural-network layer weight matrix \\
$\mathbf{w},\,w_i$ & joint and per-agent process noise \\
$\mathbf{x},\,x_i$ & joint and per-agent state \\
$x,\,y,\,z$      & inertial position coordinates, km \\
$y_{out}$        & neural-network layer output \\
$\alpha$         & extended class-$\mathcal{K}_\infty$ function \\
$\beta,\,\gamma$ & learned FiLM shift and scale parameters \\
$\gamma$         & flight-path (pitch) angle, rad \\
$\delta$         & Huber-loss transition threshold, s \\
$\Delta p,\,\Delta x$ & relative position and relative state \\
$\theta$         & neural-network parameters \\
$|\dot\theta|$   & average turning rate, rad/s \\
$\nu_{max}$      & maximum pitch rate, rad/s \\
$\sigma_i$       & process-noise (diffusion) coefficients \\
$\tau^*$         & adversarial time-to-collision (aTTC), s \\
$\hat{\tau}^*$   & predicted aTTC, s \\
$\tau^*_a$       & activation aTTC threshold, s \\
$\tau_c$         & critical aTTC threshold, s \\
$\Phi$           & flow map (solution of the dynamics) \\
$\psi$           & yaw (heading) angle, rad \\
$\omega_{max}$   & maximum yaw rate, rad/s \\
$\mathcal{C},\,\mathcal{S}$ & unsafe and safe sets \\
$\mathcal{X},\,\mathcal{U},\,\mathcal{W}$ & state, control, and disturbance spaces \\
$\odot$          & element-wise (Hadamard) product \\
\multicolumn{2}{@{}l}{Subscripts}\\
$c$              & critical \\
$col$            & collision \\
$i,\,j$          & agent indices \\
$max$            & maximum \\
$nm$             & near-miss \\
$nom$            & nominal \\
\multicolumn{2}{@{}l}{Superscripts}\\
$e$              & evader \\
$p$              & pursuer \\
$*$              & optimal / collision-defining value \\
\end{longtable*}}

\section{Introduction} 
Collision avoidance in multi-agent autonomous systems is critical in many 
applications, including robots traversing complex environments, autonomous vehicles navigating city traffic, and aerial vehicles maneuvering within an airspace. 
Avoiding collision with other moving agents, a dynamic and challenging task, often relies on the agent's ability to sense both the position and velocity of other agents and proceed towards the objective in a manner which minimizes the chance of collision. 
One approach to ensuring collision avoidance is to use control barrier functions (CBFs), which enforces safety through a forward-invariance condition that can be encoded as a linear constraint on the control input in a quadratic program (QP). The domain of safety is defined as all states where the barrier function -- usually a threshold on the Euclidean distance -- is positive. Forward-invariance then works by placing a constraint on the rate-of-change of the barrier function, ensuring that if an agent is currently safe it will not move to a region where the barrier function crosses zero \cite{ames_control_2017,ames_control_2019,wang2017safety, zhang2025gcbf+}. Another approach is the velocity obstacle (VO), which takes a more geometric approach and defines for a given ego agent and a moving obstacle the set of ego velocity vectors which would give rise to a collision if both maintained constant velocity for some time \cite{fiorini_motion_1998}. These ideas are synergistic and have recently been combined in so called Velocity Obstacle CBFs (VO-CBFs) which construct the barrier function in terms of the VO geometry -- rather than standard distance-based formulation \cite{huang2025dynamic,roncero2025multi}. Working directly in velocity space streamlines the CBF formulation -- by allowing a first order definition --  while embedding the VO in the CBF-QP framework allows for less conservative interventions as compared to the basic VO formulation. Other CBF alternatives have been proposed including the future-focused CBF (ff-CBF) \cite{black_future-focused_2022} which guarantees safety over arbitrarily long time horizons and the Collision-Cone-CBF (C3BF) which ensures that the velocity of the ego agent relative to the obstacle always points away from the line-of-sight cone to the obstacle \cite{tayal2026collision}. 

In many situations the obstacles encountered by an agent are not simply neutral moving objects, but other autonomous agents that may be cooperative or adversarial. 
The case where all other agents are assumed to be cooperative and jointly seek to avoid collision is known as reciprocal collision avoidance. In this framework the burden of avoiding collision between any two pairs of agents can be split such that each agent only executes half of the necessary evasive maneuver \cite{van_den_berg_reciprocal_2008}. 
Despite the cooperative nature, it has been shown that this idea can be implemented through a highly scalable linear program without resorting to a centralized control algorithm as long as each agent \textit{assumes} that all others obey the same control and responsibility-sharing logic \cite{van_den_berg_reciprocal_2011}. 
Adversarial collision avoidance is the opposite and limiting case where the ego agent(s) assume that any other agent is \textit{actively seeking} collision -- a concept familiar from game theory \cite{isaacs_differential_1969,weintraub_introduction_2020}.
Recently this game theoretic approach has been applied to UAV collision avoidance in the scenario where an evader with finite turning radius attempts to evade one or many infinitely-agile pursuer \cite{exarchos_uav_2016}. In related work, CBFs have been applied to a similar pursuit-evasion scenario where two or more pursuers are avoiding collision among themselves in pursuit of a common target \cite{lv_control_2024}. However, these studies were restricted to very short duration missions. Therefore, they could not provide a statistical analysis of the long term collision avoidance capability. Additionally, they were restricted to 2D and used a ballistic (constant-velocity) predictor to reason about future collision risk. These assumptions rule out vertical motion and speed modulation as evasive tactics -- both of which are critical in aerial settings.   

In this work, we propose a temporal control barrier formulation where the barrier function is neither defined in space nor velocity, but directly in time. 
To this end we build on the concept of time-to-collision (TTC), a standard metric in automotive safety \cite{hayward_near-miss_1972,barhoumi_formal_2026}. 
At any instant in time, the TTC is the time to collision between any two agents were they to continue on their current trajectory at their current speed. We argue that time is, in fact, the natural dimension in which to measure collision risk. 
Consider that when negotiating rush-hour traffic, ski slopes, or a crowded concert venue we generally navigate by an intuitive sense of our reaction time rather than an explicit evaluation of relative distances and speeds. 
In fact, this concept has been applied to collision avoidance with static obstacles \cite{bosnak2017efficient}, as well as integrated with the previously mentioned VO-CBF framework \cite{roncero2025multi}. Most recently, \cite{fu_self-triggered_2026} demonstrate that letting the barrier's amplitude scale with TTC — rather than distance alone — reduces conservatism for a single Euler-Lagrange robot avoiding obstacles of known dynamics. Generally, time-to-collision metrics assume constant velocity which does not guarantee that the TTC exists -- a problem if it is to be used for practical control purposes. In automotive or robotics settings where motion is restricted only to 2 dimensions -- or even to quantized lanes in the case of automotive control -- this is generally not a critical limitation, but in 3D aerial environments more specificity is needed to ensure a finite, and thus useful, metric.

In this work, our contributions are as follows.
\begin{enumerate}
    \item We derive a \textit{worst-case} extension of TTC which we term adversarial TTC (aTTC), where it is assumed that one agent drives towards collision under maximum control effort. Under mild assumptions, this metric is always finite and thus serves as a more natural fit for real-time control in 3D environments.
    \item We build a lightweight, differentiable surrogate model for our aTTC measure, which can be directly integrated with a CBF-QP control law and allows us to avoid costly integration of the system dynamics.
    \item We demonstrate our framework on multi-agent aerial vehicles in both independent and collaborative (formation-flight) pursuit evasion scenarios. Through long time-horizon simulations we are able to capture a statistical performance evaluation and find that our approach allows for lower collisions, greater mission progress, and tighter formation coherence as compared to a distance-based CBF formulation.
\end{enumerate}

The rest of the paper is organized as follows. The problem statement and adversarial-time-to-collision framework is introduced in \S\ref{sec:math}, the integration of which into practical CBF based control scheme is described in \S\ref{sec:ttc_cbf}. We then present our results in \S\ref{sec:results} and a discussion of the significance, limitations thereof as well as future directions in \S\ref{sec:conclusion}

\section{Mathematical Formulation}\label{sec:math}
\subsection{Problem Formulation}\label{sec:prob}
Consider a multi-agent system, each of whose $N_a$ agents $i \in \mathcal{I}$ are modeled according to the following control-affine dynamics
\begin{equation}\label{eq:agent-dynamics}
    \dot x_i = f_i(x_i) + g_i(x_i)u_i + w_i, \quad x_i(0) = x_{i,0},
\end{equation}
where $x_i(t) = (p_i(t),\, o_i(t)) \in \mathcal{X} \subseteq \mathbb{R}^{n}$, denotes the state consisting of position $p_i(t) \in \mathbb{R}^{d}$ and additional internal variables $o_i(t) \in \mathbb{R}^{q}$ such that $q \geq 0$ and $n=d+q$. The control input is $u_i \in \mathcal{U} \subseteq \R^m$, with $\mathcal{U}$ compact and convex, and $w_i \in \mathcal{W} \subset \R^n$ is the bounded process noise of agent~$i$. 
The drift field $f_i\colon \mathcal{X} \to \R^n$ and input matrix $g_i\colon \mathcal{X} \to \R^{n \times m}$ are locally Lipschitz on~$\mathcal{X}$. 
Stacking the individual states, inputs, and disturbances yields the joint system
\begin{equation}\label{eq:joint-dynamics}
    \dot{\mathbf{x}} = F(\mathbf{x}) + G(\mathbf{x})\,\mathbf{u}+ \mathbf{w}, \quad x(0) = x_0,
\end{equation}
where $\mathbf{x} = (x_1,\dots,x_{N_a}) \in \mathcal{X}^{N_a}$, $\mathbf{u} = (u_1,\dots,u_{N_a}) \in \mathcal{U}^{N_a}$, $\mathbf{w} = (w_1,\dots,w_{N_a}) \in \mathcal{W}^{N_a}$, and
\[
    F(\mathbf{x}) = \begin{bmatrix} f_1(x_1) \\ \vdots \\ f_{N_a}(x_{N_a}) \end{bmatrix}, \,
   G(\mathbf{x}) = \begin{bmatrix} g_1(x_1) & & \\ & \ddots & \\ & & g_{N_a}(x_{N_a}) \end{bmatrix},
\]
such that for a locally Lipschitz control law $\mathbf{u} = k(\mathbf{x})$, $k: \mathcal{X}^{N_a} \to \mathcal{U}^{N_a}$, the system~\eqref{eq:joint-dynamics} admits a unique solution, which shall be denoted at time $t \geq 0$ by $\Phi(\mathbf{x}, \mathbf{u}, \mathbf{w}, t)$. 
By a slight abuse of notation, $\Phi_i$ will be used to denote the solution to \eqref{eq:agent-dynamics} for agent $i$. 
We assume that some subset of agents $\mathcal{I}_e \subset \mathcal{I}$, $|\mathcal{I}_e| = N_e$, seek a nominal objective specified by a nominal control vector $\mathbf{u}_0\in \Uset^{N_e}$ and attempt to evade collisions with the others. We refer to these agents as \textit{evaders}. 
The remaining subset of agents, $\mathcal{I}_p = \mathcal{I} \setminus \mathcal{I}_e$, $|\mathcal{I}_p| = N_p$, follow some unknown, potentially adversarial control scheme. In the adversarial context these are referred to as \textit{pursuers}. 
In either case they may be considered as an external forcing to the joint system~\eqref{eq:joint-dynamics}.

Let $v_i \equiv \dot p_i$ denote the velocity and $a_i \equiv \dot v_i$ denote the acceleration of agent $i$ respectively, neither necessarily a state variable. 
In order to more accurately model aerial vehicles in practice, we assume the following.
\begin{assumption}[Velocity and Acceleration Limits]\label{ass:speed-accel-limits}
    Each vehicle is subject to speed and acceleration limits, i.e., $\forall i \in \mathcal{I}$, $|v_i(t)| \leq \bar v_i, |a_i(t)| \leq \bar a_i$ for $\bar v_i, \bar a_i > 0$, for all $t \geq 0$.
\end{assumption}

The aim of this work is twofold: first, to derive a metric quantifying the risk of future collisions that is sensitive to both the relative positions of agents and their governing dynamics; 
second, to design a control scheme for the evader agents based on the derived metric which ensures progress towards the nominal objective while avoiding collisions with pursuers and other evaders and adhering to control limits (e.g., maximum speed, acceleration, or turning rate). 
As such, we assume that each agent fills a spherical volume (area if in $\R^2$) of radius $r_{col} > 0$ -- the collision radius. A collision is defined as an event where the inter-agent distance simultaneously falls below twice this radius and is non-increasing,
\begin{equation}\label{eq:col_def}
    d_{ij}(x_i, x_j) \equiv \|p_i-p_j\|\leq 2r_{col}, \qquad \dot{d}_{ij} \leq 0.
\end{equation}
The rate condition ensures that in a discrete time setting only the moment of collision is counted as a unique event. 


\subsection{Adversarial Time-To-Collision}
The inter-agent distance $d_{ij}$ as defined by \eqref{eq:col_def} is a common metric for quantifying collision risk. 
However, such a metric is unable to distinguish between spatially close configurations that are dangerous (e.g., the agents are close and converging) versus safe (e.g., the agents are close but diverging). 
Even if the relative velocities are explicitly accounted for through a velocity-obstacle type formulation~\cite{huang2025dynamic}, it is not clear how a given separation in velocity space translates to a quantifiable risk assessment. 
In control design, we argue that the natural domain in which to quantify risk is \textit{time} -- if a danger arises, how much \textit{time} do I have to take evasive action?

Thus, to derive our risk metric we build on the concept of \textit{time-to-collision} (TTC). For two agents with given states $x_i$ and $x_j$, the TTC is defined as the elapsed time $\tau^*$ until a collision (\ref{eq:col_def}) is recorded under integration of the joint dynamics (\ref{eq:agent-dynamics}) from  $x_i$ and $x_j$ -- and some assumed control input.  
Practically, the control input is generally neglected -- in other words, the agents are assumed to move with a constant velocity vector along their current trajectory.
We now define the set of control inputs rendering the velocity vector of agent $i$ constant by $\Uset^v_i(\state_i) \triangleq \{ u \in \Uset : a_i(\state_i, u) = 0 \}$, and assume that the set $\Uset^v_i$ is never empty, formalized as follows.
\begin{assumption}[Constant-Velocity Control]\label{ass:const-vel}
     Each agent's constant-velocity control set is non-empty, i.e., $\Uset^v_i(\state_i) \neq \emptyset$ for all $\state_i \in \Xset$, for each $i \in \mathcal{I}$.
\end{assumption}

For 2D settings, such as in automotive control where the agent motion is constrained to a plane, or even to discrete lanes, assuming constant-velocity motion is reasonable. However, in a 3D aerial setting this assumption applied to each vehicle leads to a metric which is almost never finite. In such cases the agent volume $\sim r_c^3$ is generally much smaller than the size of the domain. This, in combination with the increased degree of freedom implies that the probability of two constant-velocity agents colliding is very low -- leading to the uninformative conclusion that $\tau^* = \infty$.


To derive a metric that is universally well-defined, we introduce the notion of \textit{adversarial} time-to-collision (aTTC). 

\begin{definition}[aTTC]\label{def:attc}
    Given agents $i, j \in \mathcal{I}$ with $\state_i \in \Xset_i$, $\state_j \in \Xset_j$, the \textbf{adversarial time-to-collision} (aTTC) is
    \begin{equation}\label{eq:attc}
        \tau_{ij}^* \triangleq
        \inf_{u_j \in \mathcal{M}(\Uset)} \inf\!\Big\{\tau \geq 0 \;\Big|\; d_{ij}\left(\Phi_i(x_i, \pi^v_i(x_i), \mathbf{0}, \tau), \Phi_j\!\big(x_j, u_j, \mathbf{0}, \tau\big)\right) \leq 2r_{col}\Big\},
    \end{equation}
    where $\pi^v_i$ is any measurable selection $\pi^v_i(\state_i) \in \Uset^v_i(\state_i)$, $\mathcal{M}(\Uset)$ denotes the measurable controls taking values in $\Uset$, and $\inf\emptyset = +\infty$. 
\end{definition}
\noindent That is, aTTC captures the time-to-collision between an ego agent in constant-velocity motion and a pursuer agent that pursues as aggressively as its speed and acceleration limits allow.
By Assumption~\ref{ass:const-vel}, a policy $\pi^v_i(\state_i) \in \Uset^v_i(\state_i)$ exists and the ego trajectory in \eqref{eq:attc} satisfies $\dot v_i \equiv 0$, hence
\begin{equation*}
    p_i(\tau) = p_i + \tau\, v_i, \qquad \tau \geq 0,
\end{equation*}
for every admissible selection $\pi^v_i$. 
Since $d_{ij}$ depends only on the position channel, $\tau^*_{ij}$ is independent of the selection and \eqref{eq:attc} is well defined. 
Selections may differ in the internal variables $o_i$, which do not enter the collision condition.
Critically, under mild assumptions on the maximum speed of the agents the aTTC is guaranteed to be finite.

\begin{proposition}[Finiteness of the aTTC]\label{prop:finite}
Consider agents $i,j\in\mathcal I$ governed by \eqref{eq:agent-dynamics} under Assumptions~\ref{ass:speed-accel-limits}--\ref{ass:const-vel}, and let $\mathcal A_j(x_j)\subseteq\mathbb R^{d}$ denote the set of accelerations attainable by agent $j$ at $x_j$ over $u_j\in\mathcal U$. 
Suppose that
\begin{enumerate}[label=(\roman*)]
  \item agent $j$ is strictly faster than the \emph{instantaneous} speed of
        agent $i$, i.e.\ $\bar v_j>\lVert v_i\rVert$;
  \item the acceleration of agent $j$ is freely assignable within its limit,
        $\mathcal B(0,\bar a_j)\subseteq\mathcal A_j(x_j)$ for all $x_j\in\mathcal X$.
\end{enumerate}
Then
\begin{equation}\label{eq:tau_bound}
  \tau^{*}_{ij}\;\le\;
  \frac{d_{ij}(x_i,x_j)-2r_{col}+2\bar v_j^{2}/\bar a_j}
       {\bar v_j-\lVert v_i\rVert}\;<\;\infty
\end{equation}
for every configuration $(x_i,x_j)\in\mathcal X\times\mathcal X$ with $d_{ij}\ge 2r_{col}$.
\end{proposition}

\begin{proof}
Since \eqref{eq:attc} takes the infimum over admissible controls of agent $j$, it suffices to find one admissible control attaining \eqref{eq:tau_bound}.
By Assumption~\ref{ass:const-vel} and the aTTC defined by~\eqref{eq:attc}, the ego trajectory is the ray $p_i(\tau)=p_i+\tau v_i$ with $v_i$ constant.
We organize the proof into two components: finite displacement during a finite reorientation period, and finite closing time from zero velocity.

First, apply $a_j=-\bar a_j\,v_j/\lVert v_j\rVert$ while $v_j\neq0$, admissible by (ii). 
Then $\tfrac{d}{dt}\lVert v_j\rVert=-\bar a_j$, so $v_j(t_1)=0$ for some $t_1\le\bar v_j/\bar a_j$, and the displacement incurred over $[0,t_1]$ is at most $\bar v_j^{2}/(2\bar a_j)$. 
Since agent $j$ reorients itself it may then dwell at rest for an arbitrary duration.

Next, fix a candidate time $T$ and set the aim point $y_T:=p_i+T v_i$, which is constant once $T$ is fixed. 
From rest at $p_j(t_1)$, agent $j$ applies $a_j=\bar a_j\hat u$ with $\hat u:=(y_T-p_j(t_1))/\lVert y_T-p_j(t_1)\rVert$ until $\lVert v_j\rVert=\bar v_j$, and $a_j=0$ thereafter; both are admissible by (ii) and Assumption~\ref{ass:speed-accel-limits}. 
Writing $L:=\lVert y_T-p_j(t_1)\rVert$, closing to within $2r_{col}$ of $y_T$ along this straight line takes time at most $(L-2r_{col})/\bar v_j+\bar v_j/(2\bar a_j)$.

Capture at time $T$ is therefore achievable whenever
\begin{equation}\label{eq:T_ineq}
  T\;\ge\;\underline T(T)\;:=\;
  \frac{\bar v_j}{\bar a_j}
  +\frac{\bar v_j}{2\bar a_j}
  +\frac{L-2r_{col}}{\bar v_j}.
\end{equation}
By the triangle inequality and the displacement bound $\bar v_j^2 / (2\bar a_j)$, $L\le T\lVert v_i\rVert+d_{ij}+\bar v_j^{2}/(2\bar a_j)$, so
\[
  \underline T(T)\;\le\;\frac{2\bar v_j}{\bar a_j}
  +\frac{d_{ij}-2r_{col}}{\bar v_j}
  +\frac{\lVert v_i\rVert}{\bar v_j}\,T .
\]
Thus $\underline T$ is affine in $T$ with slope $\lVert v_i\rVert/\bar v_j<1$ by (i), so $T=\underline T(T)$ admits a unique solution $T^{\dagger}$ bounded above by the right-hand side of the first inequality in~\eqref{eq:tau_bound}. 
Because agent $j$ may dwell at rest for any duration, its set of achievable arrival times at $y_T$ is the interval $[\underline T(T),\infty)$; taking $T=T^{\dagger}$ it arrives at $y_{T^\dagger}$ exactly at $T^{\dagger}$, at which instant
$\lVert p_j-p_i\rVert\le2r_{col}$. Hence $\tau^{*}_{ij}\le T^{\dagger}$.
\end{proof}

\begin{remark}
The initial reorientation phase uses $v_j=0$, which is admissible under Assumption~\ref{ass:speed-accel-limits} and for the model of \S\ref{sec:results}. 
For airframes with a positive stall speed the phase must be replaced by a minimum-radius heading reversal, which inflates the constant multiplying $\bar v_j^{2}/\bar a_j$ but leaves the structure of \eqref{eq:tau_bound} unchanged.
At the parameters of \S\ref{sec:results} ($\bar v_j=0.75$~km/s, $\bar a_j=0.05$~km/s$^2$, $\lVert v_i\rVert=0.25$~km/s, $d_{ij}=5$~km) the bound evaluates to $\approx 66$~s.
\end{remark}

The proposition guarantees that the aTTC is well-defined and finite whenever a
pursuer's speed bound exceeds the instantaneous speed of a nonreactive agent. We
stress that $\tau^*_{ij}$ is a one-sided reachability time rather than the value of a
differential game: agent $j$ optimises while agent $i$ is held non-reactive. The
metric is therefore conservative in the pursuer's assumed intent and optimistic in
the ego's assumed response, in contrast to the saddle-point formulations
of~\cite{exarchos_uav_2016}. The asymmetry is deliberate --- it is what makes the
metric computable pairwise and composable across many agents --- but it means
$\tau^*_{ij}$ lower-bounds the time available to a manoeuvring ego rather than
characterising the outcome of the encounter.
Conceptually, the aTTC quantifies the earliest time an agent would collide with the nonreactive agent under worst-case dynamics, not the agent's actual intent. 
As such, aTTC serves as a robust metric of adversarial risk.
Direct computation of the aTTC requires solving an optimal control problem and integrating the system dynamics forward from the current state to predicted collision, a computationally demanding task in all but the simplest models. 
Furthermore, integration into an optimization based control law requires a differentiable representation of the mapping from the system state to the aTTC. 
We detail how these computational challenges are addressed in Section~\ref{sec:ttc_cbf}.

\section{Time-To-Collision Barrier Functions}\label{sec:ttc_cbf}
We aim to design a control scheme $\mathbf{u}_e \in \Uset^{N_e}$ for the evader agents which maximizes the progress towards the nominal objective while minimizing collisions with pursuers and other evaders. 
For this purpose we utilize the framework of  the \textit{Control Barrier Function} (CBF) \cite{ames_control_2017}.
\subsection{Control Barrier Functions}
Let $h: \mathcal{X}^{N_a} \to \R$ be a continuously differentiable function, with $0$ a regular value, encoding the unsafe and safe configurations as $\mathcal{C}$ and $\mathcal{S}$ respectively, such that
\begin{align}
    \mathcal{C} &= \{\mathbf{x} \in \mathcal{X}^{N_a} \mid h(\mathbf{x}) < 0\}, \label{eq:collision-set} \\
    \mathcal{S} &= \{\mathbf{x} \in \mathcal{X}^{N_a} \mid h(\mathbf{x}) \geq 0\}, \label{eq:safe-set}
\end{align}
where the boundary is defined by $\partial \mathcal{S} = \{\mathbf{x} \in \mathcal{X}^{N_a} \mid h(\mathbf{x}) = 0\}$.
The goal is to enforce that $\mathbf{x}(t) \in \mathcal{S}$ on the maximal interval of existence of the closed-loop solution. 
This holds whenever $\mathcal{S}$ is \textit{forward-invariant} under~\eqref{eq:joint-dynamics}. 
Because~\eqref{eq:joint-dynamics} is driven by the bounded process noise $\mathbf{w}$, the closed loop is a differential inclusion. Thus, invariance must hold over all admissible disturbance realizations $\mathbf{w} \in \mathcal{W}^{N_a}$.
The corresponding necessary and sufficient condition is a perturbed form of Nagumo's theorem~\cite{blanchini1999set}, as follows.
\begin{lemma}\label{lem:nagumo}
    Let $h$ be continuously differentiable with $0$ a regular value, let $k\colon \Xset^{N_a} \to \Uset^{N_a}$ be locally Lipschitz, and let $\mathcal{W}$ be compact.
    Then $\Sset$ is robustly forward-invariant for
    \begin{equation}\label{eq:closed-loop-disturbed}
        \dot{\mathbf{x}} = F(\mathbf{x}) + G(\mathbf{x})k(\mathbf{x}) + \mathbf{w}(t),
        \qquad \mathbf{w}(t) \in \mathcal{W}^{N_a},
    \end{equation}
    i.e., every solution with $\mathbf{x}(0) \in \Sset$ remains in $\Sset$ on its maximal
    interval of existence for every admissible disturbance realization, if and only if
    \begin{equation}\label{eq:forward-invariance}
        L_F h(\mathbf{x}) + L_G h(\mathbf{x})k(\mathbf{x}) - \sigma_w(\mathbf{x}) \geq 0,
        \qquad \forall\, \mathbf{x} \in \partial \Sset,
    \end{equation}
    where
    \begin{equation}\label{eq:sigma_w}
        \sigma_w(\mathbf{x}) \triangleq \sup_{\mathbf{w} \in \mathcal{W}^{N_a}}
        \big[-\nabla h(\mathbf{x})^\top \mathbf{w}\big].
    \end{equation}
\end{lemma}

\noindent The supremum in~\eqref{eq:sigma_w} exists and is finite due to compactness of $\mathcal{W}^{N_a}$, and the $\sigma_w$ term is essential for invariance.
The tangency condition must hold for \emph{every} admissible disturbance, and since $\nabla h^\top \mathbf{w}$ is linear in $\mathbf{w}$ the worst case is the support function $\sigma_w$ defined by \eqref{eq:sigma_w}. 
One tool for enforcing forward invariance in control design is the control barrier function.
Partition the joint input as $\mathbf{u} = (\mathbf{u}_e, \mathbf{u}_p)$ with $\mathbf{u}_e \in \Uset^{N_e}$, $\mathbf{u}_p \in \Uset^{N_p}$, and correspondingly $G(\mathbf{x}) = [\,G_e(\mathbf{x})\;\; G_p(\mathbf{x})\,]$.

\begin{definition}[Adversarial CBF]\label{def:cbf}
    Given $\Sset \subset \mathcal{D} \subset \Xset^{N_a}$ defined by \eqref{eq:safe-set}
    for a continuously differentiable $h\colon \mathcal{X}^{N_a} \to \R$ with $0$ a regular value,
    $h$ is an \textbf{adversarial control barrier function} on $\mathcal{D}$ if there
    exists an extended class-$\mathcal{K}_\infty$ function $\alpha$ such that
    \begin{equation}\label{eq:cbf-condition}
        \sup_{\mathbf{u}_e \in \Uset^{N_e}}
        \Big[ L_F h + L_{G_e} h\, \mathbf{u}_e
        - \sigma_p(\mathbf{x}) - \sigma_w(\mathbf{x}) \Big]
        \;\geq\; -\alpha\big(h(\mathbf{x})\big),
        \quad \forall\, \mathbf{x} \in \mathcal{D},
    \end{equation}
    where
    \begin{equation}\label{eq:support-fns}
        \sigma_p(\mathbf{x}) \triangleq \sup_{\mathbf{u}_p \in \Uset^{N_p}}
        \big[-L_{G_p} h(\mathbf{x})\,\mathbf{u}_p\big],
    \end{equation}
    and $\sigma_w$ is defined as in~\eqref{eq:sigma_w}.
\end{definition}
Note that the supremum in~\eqref{eq:support-fns} is also finite and attained when $\Uset$ is compact, and that the terms $\sigma_p$ and $\sigma_w$ preserve the control-affine nature of the inequality \eqref{eq:cbf-condition} in $\mathbf{u}_e$. 
Satisfying~\eqref{eq:cbf-condition} certifies safety against every admissible pursuer input and disturbance simultaneously, while the disturbance-free CBF condition introduced by \cite{ames2016control} is recovered when $\mathcal{I}_p = \emptyset$ and $\mathcal{W} = \{0\}$.
\begin{theorem}[Robust safety]\label{thm:safety}
    Let $h$ be an adversarial CBF on $\mathcal{D} \supseteq \Sset$ per Definition~\ref{def:cbf} with $\alpha$ locally Lipschitz, and let $k_e\colon \mathcal{D} \to \Uset^{N_e}$ be any locally Lipschitz function satisfying
    \begin{equation}\label{eq:kcbf}
        k_e(\mathbf{x}) \in K_{\mathrm{cbf}}(\mathbf{x}) \triangleq
        \big\{ \mathbf{u}_e \in \Uset^{N_e} :
        L_F h(\mathbf{x}) + L_{G_e} h(\mathbf{x})\, \mathbf{u}_e - \sigma_p(\mathbf{x}) - \sigma_w(\mathbf{x}) \geq -\alpha(h(\mathbf{x})) \big\}.
    \end{equation}
    Then $\Sset$ is robustly forward-invariant for the closed loop under \emph{every} measurable pursuer input $\mathbf{u}_p$ and disturbance $\mathbf{w}$.
\end{theorem}

\begin{proof}
    The set $K_{\mathrm{cbf}}(\mathbf{x})$ is nonempty by Definition~\ref{def:cbf} and compactness of $\Uset$. Along any solution,
    \begin{align*}
        \dot h &= L_F h(\mathbf{x}(t)) + L_{G_e}h(\mathbf{x}(t))\,k_e(\mathbf{x}(t)) + L_{G_p}h(\mathbf{x}(t))\,\mathbf{u}_p(t) + \nabla h^\top \mathbf{w}(t) \\
        &\geq\; L_F h(\mathbf{x}(t)) + L_{G_e}h(\mathbf{x}(t))\,k_e(\mathbf{x}(t)) - \sigma_p(\mathbf{x}(t)) - \sigma_w(\mathbf{x}(t)) \;\geq\; -\alpha(h(\mathbf{x}(t))),
    \end{align*}
    by definition of $\sigma_w$ and $\sigma_p$ in~\eqref{eq:sigma_w} and~\eqref{eq:support-fns} respectively. 
    Since $\alpha$ is locally Lipschitz, the Comparison Lemma~\cite{khalil2002nonlinear} applied to $\dot y = -\alpha(y)$ with $y(0) = h(\mathbf{x}_0) \geq 0$ gives $h(\mathbf{x}(t)) \geq y(t) \geq 0$. 
    On $\partial\Sset$ we have $h = 0$ and $\alpha(0) = 0$, so~\eqref{eq:kcbf} reduces
    to~\eqref{eq:forward-invariance} and Lemma~\ref{lem:nagumo} applies.
\end{proof}

Because the joint dynamics~\eqref{eq:joint-dynamics} inherit a block-diagonal structure from the decoupled agent dynamics~\eqref{eq:agent-dynamics}, the Lie derivatives in~\eqref{eq:cbf-condition} decompose as
\begin{equation*}
    L_F h(\mathbf{x}) = \sum_{i=1}^{N_a} \frac{\partial h}{\partial x_i} f_i(x_i), \quad
    L_{G_e} h(\mathbf{x})\mathbf{u} = \sum_{i=1}^{N_e} \frac{\partial h}{\partial x_i} g_i(x_i)\,u_i, \quad
    L_{G_p} h(\mathbf{x})\mathbf{u} = \sum_{i=N_e+1}^{N_a} \frac{\partial h}{\partial x_i} g_i(x_i)\,u_i,
\end{equation*}
so that the CBF condition can be evaluated and enforced using each agent's local dynamics, while the safety guarantee remains rigorously defined at the global level.

\subsection{Adversarial Time-To-Collision Barrier Functions}
For collision avoidance, an often-used control barrier function $h(\mathbf{x})$ is the (square of the) inter-agent distance, 
\begin{equation}\label{eq:hocbf_barrier}
    h(\mathbf{x}) = d^2_{ij}(x_i, x_j) - 4r_c^2 
\end{equation}
with critical safety radius $r_c\geq r_{col}$.
For systems for which the function $h$ has relative-degree greater than one with respect to the system dynamics, i.e., those for which the function $h$ must be differentiated more than once before a control input appears explicitly, high-order CBFs (HOCBF)~\cite{xiao2021high} are required for control design.

We instead propose a barrier function which is defined directly in the space of the aTTC between agents. 
Let $\mathcal{P} \subseteq \mathcal{I}_e \times \mathcal{I}$ denote the set of monitored agent pairs and define, for each $(i,j) \in \mathcal{P}$,
\begin{equation}\label{eq:attc_barrier}
    h_{ij}(\mathbf{x}) = \tau^*_{ij}(x_i, x_j) - \tau_c,
\end{equation}
where $\tau_{ij}^*$ is given by \eqref{eq:attc} and $\tau_c > 0$ is a user defined critical-time for safe deconfliction, known also as the safety threshold. 
This accounts for both the relative positions and velocities as well as the dynamics and control limits of the agents. 
Furthermore, time in lieu of distance is a natural barrier in autonomous (or crewed) flight where the critical threshold is the reaction time of the control algorithm or the human pilot. 



Throughout this section, we specialize to evader dynamics for which the velocity is a state
variable,
\begin{equation}\label{eq:ego-dynamics}
    \dot p_i = v_i, \qquad \dot v_i = f^v_i(x_i) + g^v_i(x_i) u_i,
    \qquad \mathrm{rank}\; g^v_i(x_i) = d .
\end{equation}
No structural assumption is placed on the dynamics of agent $j$, which remain those
of~\eqref{eq:agent-dynamics}. Under Assumption~\ref{ass:const-vel} the ego trajectory in
Definition~\ref{def:attc} is the ray $p_i + \tau v_i$, so the aTTC is a first-hitting time
between a moving point and the reachable tube of the pursuer. Let
\begin{equation}\label{eq:reach-set}
    \mathcal{R}_j(x_j, \tau) \triangleq
    \big\{ \Pi_p\,\Phi_j(x_j, u_j, \mathbf{0}, \tau) \;:\; u_j \in \mathcal{M}(\Uset) \big\}
    \subset \R^d
\end{equation}
denote the time-$\tau$ position reachable set of agent $j$ from $x_j$ under its own, unperturbed dynamics, where $\Pi_p\colon \Xset \to \R^d$ denotes the projection onto the position coordinates,
$\Pi_p x_i = p_i$.
For the following assumption, let 
\begin{equation}\label{eq:gap}
    \psi_{ij}(\tau; x_i, x_j) \triangleq
    \dist\!\big(p_i + \tau v_i,\; \mathcal{R}_j(x_j,\tau)\big) - 2 r_{col}
\end{equation}
denote the capture gap, the distance from agent $i$'s straight-line position at time $\tau$ to the pursuer's reachable set at that time (minus the collision diameter).



\begin{assumption}[Transversal capture]\label{as:transversal}
    At $\tau = \tau^*_{ij}$ the minimum-time trajectory of agent $j$ is unique, with terminal
    state $x^\star_j$, and the projection of $p_i + \tau^*_{ij} v_i$ onto
    $\mathcal{R}_j(x_j, \tau^*_{ij})$ is unique, denoted $y^*_{ij}$. Moreover the capture gap is
    strictly decreasing at $\tau^*_{ij}$:
    \begin{equation}\label{eq:transversality}
        \partial_\tau \psi_{ij}(\tau^*_{ij})
        = 
        \big\langle \hat e_{ij}, v_i \big\rangle
        - \max_{u \in \Uset} \big\langle \hat e_{ij},\,
        f^p_j(x_j^\star) + g^p_j(x_j^\star) u \big\rangle
        \;<\; 0,
    \end{equation}
    where the maximization over $u$ selects the worst-case, fastest-closing pursuer input, and
    $\hat e_{ij} \triangleq \big(p_i + \tau^*_{ij} v_i - y^*_{ij}\big)/2 r_{col}$
    is the unit capture direction, oriented from the pursuer toward the ego.
\end{assumption}

\noindent Condition~\eqref{eq:transversality} is a Petrov condition at the capture point: the pursuer must be able to close along $\hat e_{ij}$ strictly faster than the ego recedes. 
It is the pointwise counterpart of the speed-advantage hypothesis of Proposition~\ref{prop:finite}, imposed only where it is needed rather than globally. 
The assumption can fail at grazing captures, where the ego ray is tangent to $\partial\mathcal{R}_j$ so that $\partial_\tau\psi_{ij}=0$, or where the minimum-time trajectory or the projection $y^*_{ij}$ is non-unique. These are the configurations at which the aTTC loses its relative degree, and the transversality assumption excludes them. 
For the closed-form surrogate in \S\ref{sec:surrogate}, it holds automatically whenever the pursuer possesses a speed advantage.

\begin{proposition}[Differentiability of aTTC]\label{prop:sensitivity}
    Let Assumptions~\ref{ass:const-vel} and~\ref{as:transversal} hold at $\mathbf{x}$ with
    $\tau^*_{ij} < \infty$. Then $\tau^*_{ij}$ is continuously differentiable in a
    neighbourhood of $(p_i, v_i)$, with
    \begin{equation}\label{eq:gen-gradients}
        \nabla_{p_i} \tau^*_{ij}
        = \frac{\hat e_{ij}}{\big|\partial_\tau \psi_{ij}(\tau^*_{ij})\big|},
        \qquad
        \nabla_{v_i} \tau^*_{ij} = \tau^*_{ij}\, \nabla_{p_i} \tau^*_{ij} .
    \end{equation}
\end{proposition}

\begin{proof}
    Fix $\tau$ and write $z(\tau) = p_i + \tau v_i$. 
    By Danskin's theorem for parametric minimization~\cite[Prop.~B.25]{bertsekas1999nonlinear}, at a unique minimizer the distance in~\eqref{eq:gap} is differentiable in $z$ with gradient $(z - y^*_{ij})/\|z - y^*_{ij}\| = \hat e_{ij}$, where $y^*_{ij}$ is the projection of $p_i + \tau^*_{ij} v_i$ onto
    $\mathcal{R}_j(x_j, \tau^*_{ij})$.
    Since $\partial z/\partial p_i = I$ and
    $\partial z/\partial v_i = \tau I$, the chain rule gives
    $\nabla_{p_i}\psi_{ij} = \hat e_{ij}$ and
    $\nabla_{v_i}\psi_{ij} = \tau^*_{ij}\,\hat e_{ij}$. Assumption~\ref{as:transversal}
    supplies $\partial_\tau \psi_{ij} \neq 0$, so the implicit function theorem applied to
    $\psi_{ij}(\tau; p_i, v_i) = 0$ yields
    $\nabla \tau^*_{ij} = -\nabla \psi_{ij}/\partial_\tau \psi_{ij}$, and
    $\partial_\tau\psi_{ij} < 0$ fixes the sign.
\end{proof}

\begin{corollary}[Relative-degree one]\label{cor:rel-degree}
    Let $i \in \mathcal{I}_e$ have dynamics~\eqref{eq:ego-dynamics} and let
    $\mathbf{x} \in \partial\Sset$, so that $\tau^*_{ij} = \tau_c > 0$. Then
    \begin{equation}\label{eq:lie-attc}
        L_{g_i} h_{ij}(\mathbf{x})
        = \frac{\tau_c}{\big|\partial_\tau \psi_{ij}(\tau_c)\big|}\,
        \hat e_{ij}^{\,\top} g^v_i(x_i) \;\neq\; 0 .
    \end{equation}
    Consequently $\nabla h_{ij} \neq 0$ on $\partial\Sset$, so $0$ is a regular value of
    $h_{ij}$ as required by Definition~\ref{def:cbf} and Lemma~\ref{lem:nagumo}, and $h_{ij}$
    has relative degree one with respect to $u_i$ irrespective of the dynamics of agent $j$.
\end{corollary}

\begin{proof}
    Since $h_{ij}$ depends on $x_i$ only through $(p_i, v_i)$ and
    $\partial \dot p_i/\partial u_i = 0$ in~\eqref{eq:ego-dynamics}, the chain rule and
    Proposition~\ref{prop:sensitivity} give~\eqref{eq:lie-attc}. The prefactor is strictly
    positive, $\hat e_{ij}$ is a unit vector, and $g^v_i$ has full row rank, so the product is
    nonzero.
\end{proof}

\noindent Equation~\eqref{eq:gen-gradients} admits a direct physical interpretation, that the gradient of
the aTTC with respect to the ego velocity is parallel to the capture direction, so the barrier
constraint steers the evader's acceleration away from the point at which it would be
intercepted, with a gain inversely proportional to the closing rate at capture.

\subsection{Time-to-Collision under an Assumed Pursuit Law}\label{sec:surrogate}
The exact aTTC of Definition~\ref{def:attc} optimizes over all admissible pursuer controls and is not available in closed form. 
In practice, we commit to a representative pursuit law (e.g., a fixed feedback $\pi^P_j$ such as proportional
navigation (PN)) and evaluate the time-to-collision that law induces. 
This trades the worst-case guarantee for a metric that is computable by a single forward integration and reflects a concrete, physically realizable adversary.

\begin{definition}[Pursuit-law TTC]\label{def:pn-ttc}
    Fix a locally Lipschitz pursuit law $\pi^P_j\colon \Xset \times \Xset \to \Uset$
    and let $p_j^{P}(\tau) \triangleq \Pi_p\,\Phi_j\big(x_j, \pi^P_j(\cdot\,,x_i), \mathbf{0}, \tau\big)$
    denote the resulting pursuer position under the ego's constant-velocity motion
    $p_i(\tau) = p_i + \tau v_i$. The \textbf{pursuit-law time-to-collision} is
    \begin{equation}\label{eq:pn-ttc}
        \tau^{P}_{ij} \triangleq \inf\big\{\tau \geq 0 \;:\; \psi_{ij}(\tau) \leq 0\big\},
        \qquad
        \psi_{ij}(\tau) \triangleq \big\| p_i + \tau v_i - p_j^{P}(\tau)\big\| - 2 r_{col},
    \end{equation}
    with $\inf\emptyset = +\infty$.
\end{definition}

\noindent Unlike the exact aTTC, $\tau^{P}_{ij}$ is defined by a single trajectory rather than an optimization, so the reachable-set distance in~\eqref{eq:gap} reduces to the point-to-point gap in~\eqref{eq:pn-ttc}. 
We stress this change in status. Because $\pi^P_j$ is one particular pursuer rather than the closing-optimal one, $\tau^{P}_{ij}$ carries no worst-case ordering with respect to $\tau^*_{ij}$, and it may over- or
under-estimate the true minimum-time intercept depending on geometry. 
The barrier
\begin{equation}\label{eq:pn-barrier}
    h_{ij}^p(\mathbf{x}) = \tau^{P}_{ij}(x_i, x_j) - \tau_c
\end{equation}
therefore certifies safety \emph{against the assumed pursuit law}, a guarantee common in the guidance literature, and not against an arbitrary adversary. 
Robustness to the latter is recovered only by reinstating the reachable-set formulation of Definition~\ref{def:attc}.

Finiteness of $\tau^{P}_{ij}$ follows from classical results on the chosen law: a PN pursuer with a speed advantage and bounded lateral acceleration intercepts a constant-velocity target in finite time~\cite{zarchan2012tactical}, so $\tau^{P}_{ij} < \infty$ whenever $\bar v_j > \|v_i\|$ under standard PN gains. 



\begin{proposition}[Differentiability of the pursuit-law TTC]\label{prop:lower-bound}
    Let $\pi^P_j$ be continuously differentiable, let $\tau^{P}_{ij} < \infty$, and
    suppose the transversality condition
    \begin{equation}\label{eq:pn-transversal}
        \partial_\tau \psi_{ij}(\tau^{P}_{ij})
        = \big\langle \hat e_{ij},\, v_i - \dot p_j^{P}(\tau^{P}_{ij}) \big\rangle < 0,
        \qquad
        \hat e_{ij} \triangleq \frac{p_i + \tau^{P}_{ij} v_i - p_j^{P}(\tau^{P}_{ij})}{2 r_{col}},
    \end{equation}
    holds. Then $\tau^{P}_{ij}$ is continuously differentiable in a neighbourhood of
    $(p_i, v_i)$, with
    \begin{equation}\label{eq:pn-gradients}
        \nabla_{p_i} \tau^{P}_{ij}
        = \frac{\hat e_{ij}}{\big|\partial_\tau \psi_{ij}(\tau^{P}_{ij})\big|},
        \qquad
        \nabla_{v_i} \tau^{P}_{ij} = \tau^{P}_{ij}\, \nabla_{p_i} \tau^{P}_{ij}.
    \end{equation}
\end{proposition}

\begin{proof}
    Along the ego ray $z(\tau) = p_i + \tau v_i$ the gap $\psi_{ij}$
    in~\eqref{eq:pn-ttc} is a composition of the smooth flow $p_j^{P}$ and the
    Euclidean norm, hence continuously differentiable wherever $\|z - p_j^{P}\| =
    2r_{col} > 0$. The remainder follows from the proof of Proposition~\ref{prop:sensitivity}.
\end{proof}

\noindent The single-trajectory setting sharpens the transversality condition of
Assumption~\ref{as:transversal}: the worst-case maximization $\max_{u\in\Uset}\langle
\hat e_{ij}, \cdot\rangle$ collapses to the actual guidance-law closing velocity
$\langle \hat e_{ij}, \dot p_j^{P}\rangle$, and the unique-projection hypothesis is
vacuous because a single point replaces the reachable set. Danskin's theorem is
correspondingly not required; \eqref{eq:pn-gradients} follows from the implicit
function theorem alone. The gradient again lies along the line of sight $\hat e_{ij}$,
so the barrier steers the evader's acceleration directly away from the predicted
intercept, with a gain set by the closing rate at capture.

\begin{corollary}[Relative degree one]\label{cor:pn-rel-degree}
    Under the hypotheses of Proposition~\ref{prop:lower-bound}, if $i \in \mathcal{I}_e$
    has dynamics~\eqref{eq:ego-dynamics}, then $L_{g_i} h_{ij}^p(\mathbf{x}) =
    |\partial_\tau\psi_{ij}|^{-1}\,\hat e_{ij}^\top g^v_i(x_i) \neq 0$ on
    $\partial\Sset$, so $h_{ij}^p$ has relative degree one with respect to $u_i$ and the
    high-order construction of~\cite{xiao2021high} is not required.
\end{corollary}

\begin{proof}
    Immediate from~\eqref{eq:pn-gradients}, $\partial\dot p_i/\partial u_i = 0$
    in~\eqref{eq:ego-dynamics}, unit $\hat e_{ij}$, and full row rank of $g^v_i$.
\end{proof}
Because $p_j^{P}$ depends on the pursuer's full state, including $v_j$, the barrier~\eqref{eq:pn-barrier} has $L_{g_j} h_{ij}^p \neq 0$ in general, so the support function $\sigma_p^{ij}$ in~\eqref{eq:support-fns} does not vanish and is retained in the adversarial CBF condition~\eqref{eq:cbf-condition}.

\subsection{Control Synthesis}
For all of the TTC barrier function formulations introduced thus far in Section~\ref{sec:ttc_cbf}, we adopt the following quadratic program (QP) based control law:
\begin{subequations}\label{eq:cbf_qp_controller}
\begin{align}
    \mathbf{u}_e^*(\mathbf{x})
      &= \argmin_{\mathbf{u}_e \in \Uset^{N_e}}\;
      \tfrac{1}{2}\big\|\mathbf{u}_e - \mathbf{u}_0(\mathbf{x})\big\|_{M}^{2}, \label{subeq:objective} \\
    \textrm{s.t.}\quad
    L_F h_{ij} &+ L_{G_e} h_{ij}\, \mathbf{u}_e - \sigma_p^{ij} - \sigma_w^{ij}
      \;\geq\; -\alpha\big(h_{ij}(\mathbf{x})\big),
      \qquad \forall\, (i,j) \in \mathcal{P}, \label{subeq:cbf_constraint}
\end{align}
\end{subequations}
where~\eqref{subeq:objective} seeks to minimize the deviation from the nominal controller
$\mathbf{u}_0\colon \Xset^{N_a} \to \Uset^{N_e}$ under a weighting matrix
$M \in \R^{m \times m}$, $M \succ 0$, and
\eqref{subeq:cbf_constraint} enforces one adversarial CBF constraint per monitored pair. 

\begin{remark}[Well-posedness of the controller]\label{rem:qp}
    Problem~\eqref{eq:cbf_qp_controller} has a strictly convex objective and affine constraints, so wherever its feasible set is nonempty the minimizer $\mathbf{u}^*_e(\mathbf{x})$ is unique. 
    By Corollary~\ref{cor:rel-degree} the constraint gradients $L_{G_e}h_{ij}$ are nonzero on $\partial\Sset$; where in addition the active constraints satisfy LICQ, the solution map $\mathbf{x} \mapsto \mathbf{u}^*_e(\mathbf{x})$ is locally Lipschitz~\cite{ames_control_2019}, which is the regularity required by Lemma~\ref{lem:nagumo} and Theorem~\ref{thm:safety}, and the robust safety guarantee then holds wherever~\eqref{eq:cbf_qp_controller} is feasible. 
    Because the input set $\Uset^{N_e}$ is possibly bounded, feasibility may fail near $\partial\Sset$; characterizing the feasible region requires an input-feasibility condition relating $\tau_c$ to the control authority $\bar a_i$, which lies beyond the scope of the present work.
\end{remark}

Two barriers of the form~\eqref{eq:attc_barrier} are available for use in \eqref{subeq:cbf_constraint}: the aTTC-CBF $h_{ij}$ defined by~\eqref{eq:attc_barrier} and the pursuit-law TTC-CBF $h_{ij}^p$ defined by~\eqref{eq:pn-barrier}.
Whereas $h_{ij}$ provides strong guarantees against collisions with a maximally pursuant adversary, the form $h_{ij}^p$ circumvents the computational challenges of determining the reachable pursuer set at the cost of robustness to the pursuer's chosen control policy.
Challenges remain, however, in that integrating the system trajectories forward in time may be computationally demanding for complex dynamics models.
As such, we opt to train a fast and differentiable surrogate model to approximate the pursuit-law TTC $\tau_{ij}^P$.

We seek a parameterized mapping from the state of a given agent pair $(x_i, x_j)$ to the
pursuit-law TTC $\tau^{P}_{ij}$ of Definition~\ref{def:pn-ttc}. Let the relative position of
the two agents be $\Delta p_{ij} \equiv p_i - p_j$ and their velocities be $v_i$ and $v_j$.
Since the agents are subject to speed limits, the surrogate must be conditioned on the
absolute velocities of each agent during both training and evaluation, rather than the
relative velocity alone.
In addition to the state, $\tau^{P}_{ij}$ depends on the system parameters: the collision
radius $r_{col}$ and the pursuer's speed and acceleration bounds $\bar v_j$ and $\bar a_j$. The
strongest dependence is on $\bar v_j$, the pursuer's maximum attainable speed, which sets the
rate at which the pursuer closes on the ego; for clarity we hold $r_{col}$ and $\bar a_j$ fixed
and expose the model's parametric dependence on $\bar v_j$ alone.
We then seek a model parameterized by $\theta$,
\begin{equation}\label{eq:ml_model}
    G_{\theta}[\Delta p_{ij}, v_i, v_j, \bar v_j]\colon
    \R^d \times \R^d \times \R^d \times \R \to \R,
\end{equation}
which approximates the forward integration of the assumed pursuit law $\pi^P_j$ implicit in
Definition~\ref{def:pn-ttc}, and is differentiable in $(\Delta p_{ij}, v_i, v_j)$ so that the
barrier gradients required by~\eqref{subeq:cbf_constraint} are available in closed form.

If the surrogate is to apply across a range of pursuer types, its predictions must be
\textit{conditioned} on the assumed pursuer capability. We incorporate the parametric
dependence on $\bar v_j$ through Feature-wise Linear Modulation
(FiLM)~\cite{perez_film_2017}, which adds a parallel fully-connected branch that takes the
system parameters as input and modulates the activations of the main branch.
Specifically, for a layer with weights $W$, bias $b$, and activation $\operatorname{a}(\cdot)$,
the standard output $y_{out} = \operatorname{a}(Wx + b)$ is replaced by
\begin{equation}\label{eq:FiLM}
    \begin{aligned}
    z &= \gamma \odot (Wx + b) + \beta, \\
    y_{out} &= \operatorname{a}(z),
    \end{aligned}
\end{equation}
where $\gamma$ and $\beta$ are produced by the parallel branch and $\odot$ is the elementwise
Hadamard product. Because $\bar v_j$ acts by \textit{scaling} the time-to-collision, i.e., setting
how quickly the pursuer closes the gap to the ego, it is more naturally incorporated through
this multiplicative modulation than as an additional input to the main branch, though the
optimal choice may depend on the application and on the number of parameters exposed to the
CBF. We stress that the present study is not a demonstration of frontier deep-learning
capability but a conceptual instance of a data-driven CBF. 
Further details of the surrogate's training appear in Appendix~\ref{app:nn-training}.
 
We therefore consider the deployed barrier
\begin{equation}\label{eq:barrier}
    h_{ij}^\theta(\mathbf{x}) = G_{\theta}[\Delta p_{ij}, v_i, v_j, \bar v_j] - \tau_c,
\end{equation}
\noindent where $G_\theta$ is trained to
reproduce the pursuit-law TTC $\tau^{P}_{ij}$ of Definition~\ref{def:pn-ttc} from labels
generated offline by forward integration of the assumed law $\pi^P_j$. It replaces the
online trajectory integration underlying $h^p_{ij}$ with a single network evaluation, so
that the barrier and all its gradients in~\eqref{subeq:cbf_constraint} are available at each
control step. The inputs $(\Delta p_{ij}, v_i, v_j, \bar v_j)$ are those on which
$\tau^{P}_{ij}$ depends, with the pursuer speed bound $\bar v_j$ conditioning the surrogate on
the assumed adversary capability.

We stress the resulting tradeoff. Because $G_\theta$ approximates $\tau^{P}_{ij}$ with
two-sided error, the guarantee of \S\ref{sec:surrogate} (forward invariance of
$\{h^p_{ij} \ge 0\}$ certifying $\tau^{P}_{ij} \ge \tau_c$ against the assumed pursuit
law) does not transfer to the deployed barrier $h^\theta_{ij}$. It can be recovered by
training under the one-sided constraint $G_\theta \le \tau^{P}_{ij}$, so that
$\{h^\theta_{ij} \ge 0\} \subseteq \{\tau^{P}_{ij} \ge \tau_c\}$; this reinstates safety
against the assumed law at the cost of some conservatism and is a natural extension. Since
$\tau^{P}_{ij}$, and hence $G_\theta$, depends on the pursuer velocity $v_j$, the Lie
derivative $L_{g_j} h^\theta_{ij}$ is nonzero and the adversarial term
$\sigma_p^{ij}$ in~\eqref{subeq:cbf_constraint} is retained, consistent with
Corollary~\ref{cor:pn-rel-degree}.

Forward invariance is moreover conditional on feasibility of~\eqref{eq:cbf_qp_controller}: as noted in Remark~\ref{rem:qp}, bounded control authority, geometric constraints, and process noise can each render the program infeasible independent of the surrogate's accuracy, a known limitation of the controller~\eqref{eq:cbf_qp_controller} for functions $h_{ij}$ that remain unverified, i.e., candidate CBFs. 
We show in \S\ref{sec:results} that over long horizons the non-adversarial collision rate of the surrogate-based controller is on the same order as that of an exact distance-based CBF.

We note that neural networks have also been used to learn forward-invariance conditions
directly from data~\cite{dawson_safe_2021,dawson_safe_2022,yu_sequential_2023,so_how_2023},
an idea recently extended to multi-agent systems through graph neural
networks~\cite{zhang2025gcbf+}. CBFs have further been incorporated as differentiable layers
within network architectures~\cite{xiao_barriernet_2023}; such constructions are compatible
with any NN-based controller and trainable by standard gradient descent.

\section{Simulated Case Studies}\label{sec:results}

\subsection{Dynamics Model}
As a prototype model of autonomous aircraft dynamics we consider for each agent a 3D Dubins type model where the state is parameterized by the six dimensional vector
$x_i=[p^x_i,p^y_i,p^z_i,\psi_i,\gamma_i,V_i]$ corresponding to the three positional coordinates, yaw, pitch, and speed -- the latter three of which fully determine the velocity vector of the agent. The deterministic drift components \eqref{eq:joint-dynamics} take the form
\begin{subequations}\label{eq:dubins}
\begin{align}
    f_i(x_i) &= 
    V_i\begin{bmatrix}
         \cos\gamma_i \cos\psi_i, \cos\gamma_i \sin\psi_i,  \sin\gamma_i,0,0,0
    \end{bmatrix}^T, \\
    g(\mathbf{x}_i)\mathbf{u}_i &= 
    \begin{bmatrix}
        \mathbf{0}_{3\times3},\mathbf{I}_{3\times3}
    \end{bmatrix}^T[u^\psi_{i},u^\gamma_{i},u^v_{i}]^T
\end{align}
\end{subequations}
with control limits $|u^\psi_{i}| \leq \dot{\psi}_{\max}$,  $|u^\gamma_{i}| \leq {\dot \gamma_{\max}}$,  $|u^v_{i}| \leq a_{\max}$. Inspired by the aerospace application relevant to this study, our simulation is conducted in nominal units of kilometers and seconds. 
The stochastic forcing is assumed uncorrelated across dimensions and constant across
agents, and is generated as an increment of a Wiener process with diffusion
coefficients
\begin{equation}\label{eq:sigma}
    \sigma_i = \operatorname{diag}\left([0.01, 0.01, 0.01, 0.01, 0.005, 0.005]\right)
\end{equation}
in units of km\,s$^{-1/2}$, km\,s$^{-1/2}$, km\,s$^{-1/2}$, rad\,s$^{-1/2}$,
rad\,s$^{-1/2}$ and km\,s$^{-3/2}$ respectively, ordered as in $x_i$. The realised
increment is truncated componentwise at $3\sigma$ so that the disturbance $\mathbf{w}$ takes
values in a compact, convex set $\mathcal{W}$, as required by
Lemma~\ref{lem:nagumo} and Definition~\ref{def:cbf}; the support
function~\eqref{eq:sigma_w} is evaluated in closed form on this set and included
in~\eqref{subeq:cbf_constraint}.
The nominal control limits are set to $a_{\max} = 0.05$~km/s$^2$, and $\dot{\psi}_{\max} = 0.4$, and $ \dot{\gamma}_{\max} = 0.2$~rad/s and the collision radius and maximum speed are set at $r_{col} = 0.1$km and $v_{\max} = 0.5$~km/s respectively.  These, together with the noise levels above, are chosen to be broadly representative of a modern fighter-class aerial vehicle rather than calibrated to any specific airframe.

\begin{remark}\label{rem:dubins}
For the model~\eqref{eq:dubins} the velocity is $v_i = V_i\,e_i$ with $ e_i = (\cos\gamma_i\cos\psi_i,\;\cos\gamma_i\sin\psi_i,\;\sin\gamma_i), \quad
  e_{\psi,i} = (-\sin\psi_i,\;\cos\psi_i,\;0), \quad
  e_{\gamma,i} = (-\sin\gamma_i\cos\psi_i,\;-\sin\gamma_i\sin\psi_i,\;\cos\gamma_i)$,
so that $\dot v_i = g^v_i(x_i) u_i$ with
$g^v_i(x_i) = \big[\, V_i\cos\gamma_i\, e_{\psi,i} \ \ V_i\, e_{\gamma,i} \ \ e_i \,\big]$.
The triad $\{e_{\psi,i}, e_{\gamma,i}, e_i\}$ is orthonormal, hence
$\operatorname{rank} g^v_i = 3 = d$ whenever $V_i > 0$ and $\gamma_i \neq \pm\pi/2$: the
condition~\eqref{eq:ego-dynamics} and Corollary~\ref{cor:rel-degree}
apply away from the hover and vertical-flight singularities. The attainable
acceleration set is the state-dependent box
\[
  \mathcal{A}_i(x_i) = \big\{\, V_i\cos\gamma_i\, u^\psi_i e_{\psi,i} + V_i\, u^\gamma_i e_{\gamma,i}
  + u^v_i e_i \;:\; |u^\psi_i| \le \dot\psi_{\max},\;
  |u^\gamma_i| \le \dot\gamma_{\max},\; |u^v_i| \le a_{\max} \,\big\},
\]
which contains $\mathcal{B}(0, \bar a_i)$ with
$\bar a_i = \min\{v\cos\gamma\,\dot\psi_{\max},\; v\,\dot\gamma_{\max},\; a_{\max}\}$.
At the nominal cruise condition $v = 0.25$~km/s, $\gamma = 0$ this gives
$\bar a_i = 0.05$~km/s$^2$, and at the fast-pursuer speed $v = 0.75$~km/s it gives
$\bar a_j = 0.05$~km/s$^2$, so hypothesis~(ii) of Proposition~\ref{prop:finite} holds
with the nominal acceleration limit at both operating points. Acceleration authority degrades as $v \to 0$ or $|\gamma| \to \pi/2$.
\end{remark}

\subsection{Neural Network Training}
To generate the training data we construct a series of scenarios where two agents navigate 3D space, e.g., follow randomly spaced way-points or fly at each other in a head on configuration. The various scenarios are chosen to expose the NN to a wide array of relative positions and velocities. 
\textcolor{black}{The simulations are run for 50,000 simulated seconds at a time step of 0.1s and repeated four times each of which with the agents flying at a different nominal speeds: $v_{nom} = [0.15,0.25,0.35,0.5]$km/s -- the evader max speed is fixed at $0.5$ km/s for all simulations. For each of these trajectories we then compute the aTTC  at five distinct values of $v_{max}= [0.1,0.3,0.5,0.7,0.9]$ km/s values each of which leads to an independent set of paired training data: $[\Delta p,v_i, v_j,v_{max}]\rightarrow \tau^*(v_{max})$. This leads to a total training data set of $50,000\times 10\times 4\times 5 = 10,000,000$ samples.
The wide spread of scenarios helps expose the NN to a wide variety of relative speeds and helps calibrate the FiLM modulation branch of the architecture. This in turn allows the model to be applicable to a wide range of ego agent speeds and assumed pursuer speeds. The data set is randomly shuffled and split into a training portion: 75\% and a testing portion: 25\%. The NN architecture is a simple feed-forward network. The main branch has 3 fully connected layers (sizes: 128-128-64) which maps $\Delta p,v_i, v_j \rightarrow \tau^*$. The FiLM branch has layer sizes 16-16-640. The final size in the FiLM branch being determined by the total number of parameters needed to apply the modulation \eqref{eq:FiLM} to the output of each main branch layer.}

\textcolor{black}{To illustrate the performance of the model as a function of the time horizon we show in Figure~\ref{fig:ml_cond_mae} the conditional MAE and MPE (mean percentage error) for both the training and testing data -- over all values of $v_{max}$. The conditional error is defined as a function of the threshold $T$ as the error for all samples where the true $\tau^* < T$. We show both the error when the NN is evaluated in samples seen in training (circles) and unseen testing data (solid lines). We highlight three aspects of this figure. First, the error on seen and unseen data is nearly identical highlighting that our model successfully generalizes. Second, in the most relevant regime: $T < ~30s$ the model maintains a MAE of less than 1.5s -- a reasonable level of accuracy for our purposes. Three, the MPE (approximately) plateaus for large values of $T>~60$s indicating that our model retains relative accuracy for safe configurations where $\tau^*$ is large. Next, we demonstrate the performance of the model as a function of $v_{max}$. Figure~\ref{fig:ml_error_pdf} shows the pdf of the raw error for the critical regime of $\tau^*<30$s for four distinct values of $v_{max}$: 0.3 (slower than the evader), 0.5 (as fast as the evader), 0.6, and 0.9 (both faster than the evader). The mean (bias) and skewness of the error is listed in the legend. Consistent with Proposition~\ref{prop:finite} we observe very low error ($<1$s) for all cases where the pursuer is \textit{at least as fast} as the evader, while for the slower pursuer case we observe slightly higher biases ($~2$s). This is due to the training data having significantly fewer samples to learn from and the samples tending to higher values of $\tau^*$ -- slower pursuers take longer to catch the target. Additionally, we see that for all cases the bias and skewness are negative, indicating the model is conservative or tends to underpredict -- a generally more favorable bias for collision avoidance. This is also evidenced by the noticeably heavier negative tail of the pdf.}

\begin{figure}
    \centering
    \begin{subfigure}[t]{0.49\textwidth}
        \centering
        \includegraphics[width=\textwidth]{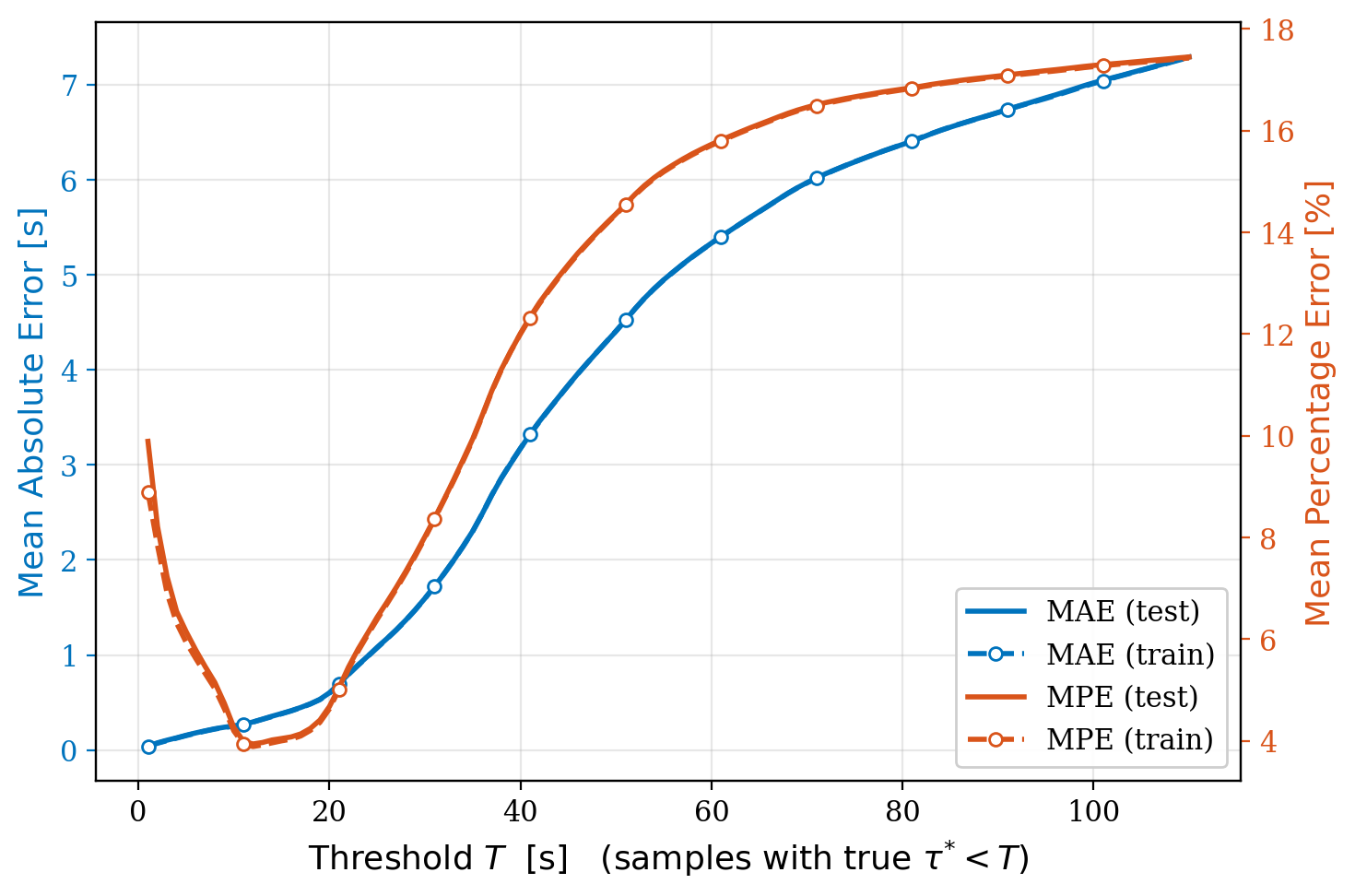}
        \caption{Conditional MAE (blue/left-axis) and conditional MPE (red/right-axis) of NN surrogate applied to training data (dashed lines/circles) and unseen test data (solid lines).}
        \label{fig:ml_cond_mae}
    \end{subfigure}
    \hfill
    \begin{subfigure}[t]{0.49\textwidth}
        \centering
        \includegraphics[width=\textwidth]{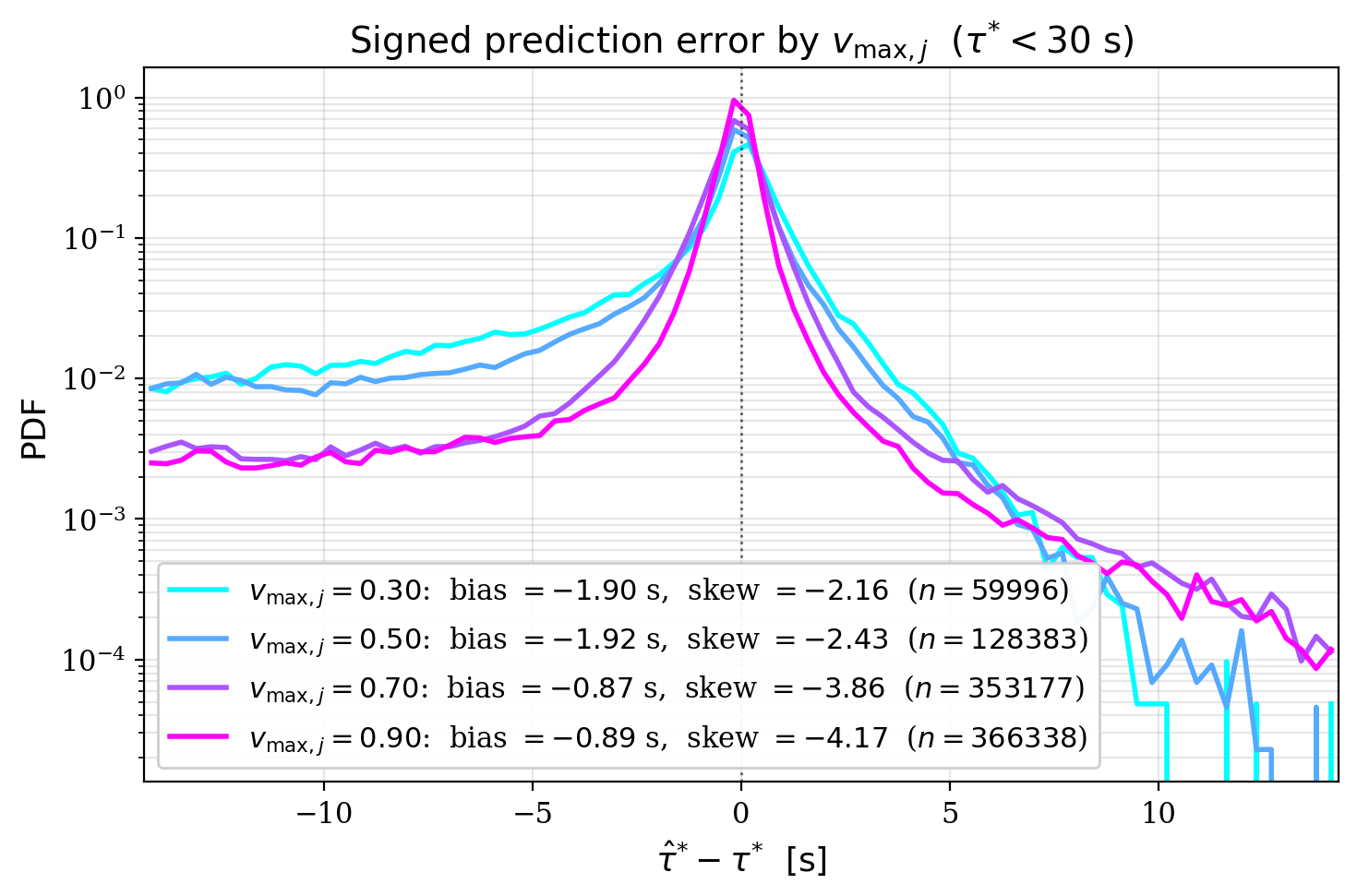}
        \caption{Probability density functions of NN prediction error for various assumed values of $v_{max}$. Legend shows mean bias and skewness values - negative values indicate under-prediction, positive values indicate over-prediction. (units are in km/s and s)}
        \label{fig:ml_error_pdf}
    \end{subfigure}
    \caption{Error metrics for NN surrogate model.}
    \label{fig:ml_combined}
\end{figure}


\subsection{Numerical Experiment}
As a numerical experiment we consider the stochastically forced 3D Dubins model~\eqref{eq:dubins}. 
The agents are split into \textit{evaders} which attempt to achieve a nominal objective (follow waypoints) and \textit{pursuers} which attempt to collide with the evaders. All agents use proportional heading and speed control as their nominal inputs $\mathbf{u}_0$, in other words $u_{\alpha, i}^0 = K_{\alpha}(\alpha_{des.}-\alpha_{act.}), ~\ \alpha = \psi_i,\gamma_i,V_i$. For the heading angles ($\psi_i, \gamma_i$) the desired heading corresponds to the direction of the next waypoint, for the speed, the desired speed is a fixed nominal speed. 
The evaders travel at a nominal speed of $V^e_{nom}=0.25$km/s and are capable of accelerating to a maximum speed of $V^e_{max} = 0.5$km/s to avoid collision. The pursuers attempt to collide with the evaders through the same proportional control scheme that navigates towards a far off point collinear with the pursuer and agent -- as illustrated in Figure~\ref{fig:missions}c. This ensures that the pursuer attempts to intercept the evader at a maximum speed $V^p_{max}$.

\textcolor{black}{We consider multiple scenarios: \textit{independent multi-agent pursuit evasion} where multiple evaders pursue independent objectives (waypoints) both without and while evading multiple pursuers and \textit{formation flight pursuit evasion} where a group of evaders in formation avoid pursuers while following the same objectives.} Figure~\ref{fig:missions} shows a schematic diagram of the two mission scenarios. In order to obtain a statistical view of the collision avoidance and waypoint tracking skill we run all simulations for $T_{sim} = 50,000$s. \textcolor{black}{In order to illustrate the benefits of the aTTC-CBF, we compare our method to a higher-order CBF (HOCBF) -- both the aTTC-CBF and the HOCBF have a similar per-call overhead of less than 1.0 ms in Matlab on a CPU. We note that this comparison is \textit{not} intended to demonstrate the superiority of our approach to any \textit{specific} state-of-the-art method but rather to illustrate phenomenologically the advantages of time based control barrier formulation. Both the aTTC-CBF and HOCBF rely on several parameters. First, the activation scale: $\tau^*_a, r_a$, which is the aTTC or distance above which the CBF is inactive -- this is particularly relevant for formation flight where an unrestricted CBF constraint can break formation even under nominal conditions. The second is the critical scale $\tau^*_c, r_c$  defined in \eqref{eq:attc_barrier} and \eqref{eq:hocbf_barrier} which define the barrier functions themselves. Third is the class $\mathcal{K}$ function $\alpha$ used to define the forward invariance condition \eqref{eq:forward-invariance}. Furthermore, these parameters must be chosen independently for evader-evader interactions and evader-pursuer interactions.  
For both the aTTC-CBF and HOCBF we optimized these parameters to minimize collisions and maximize mission progress. }

\begin{figure}[t]
    \centering
    \includegraphics[width=0.9\linewidth]{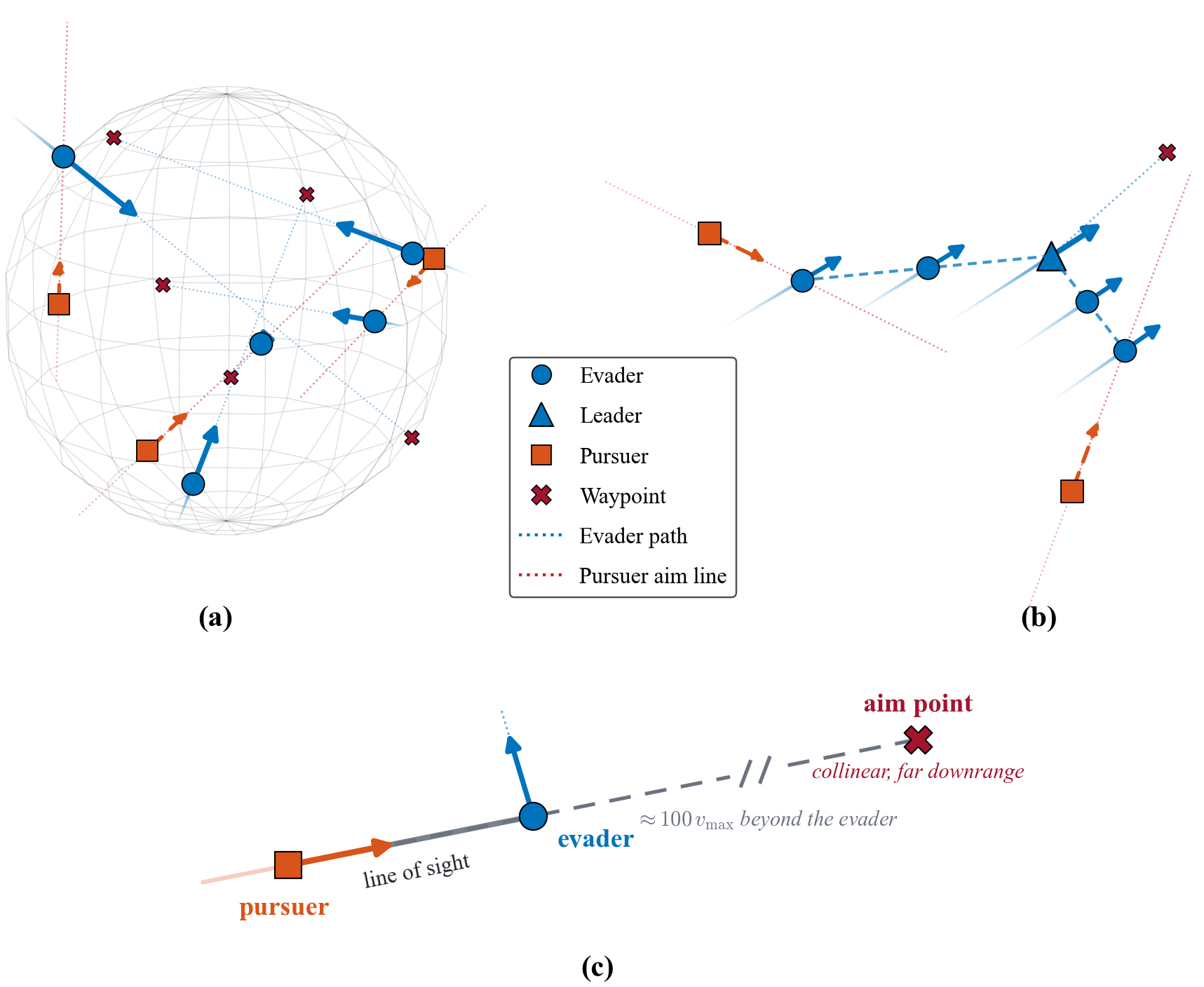}
    \caption{Schematic of the independent multi-agent (a) and formation-flight (b) pursuit--evasion scenarios and the pursuit control scheme (c). Evaders are shown in blue and pursuers in red; the leader is present only in the formation case.}
    \label{fig:missions}
\end{figure}

\subsection{Independent Multi-Agent Collision Avoidance}
As our first scenario we consider multiple agents all pursuing independent waypoints in a dense domain while avoiding collision among themselves and with any potential adversarial pursuers. The agents are initialized at random locations on a sphere of radius $R = 3.75$\,km. Each evader then navigates to a random location on the opposite hemisphere at a nominal speed of $v_{nom}^e = 0.25$\,km/s; once the waypoint is reached, a new waypoint on the opposite hemisphere is chosen, \emph{ad infinitum}. Each pursuer chases its nearest evader until collision, after which it selects a new random evader. The sphere radius is set so that the nominal travel time across the diameter is 30~s, yielding a dense, target-rich environment in which agent trajectories intersect in a variety of ways. By always placing the next waypoint on the opposite hemisphere, we ensure that evaders cannot simply exploit a speed advantage to continuously outrun their pursuers.

We consider three cases: (i) 8 independent evaders with no adversarial
pursuers; (ii) 8 evaders and 3 \emph{slower} pursuers, $V_{max}^P < V_{max}^e$; and (iii) 8 evaders and 3 \emph{faster} pursuers, $V_{max}^P > V_{max}^e$. The slower- and faster-pursuer cases use $V_{max}^P/V_{max}^e = 0.9$ and $1.5$, respectively. We compare three controllers: the no-CBF baseline, the reference HOCBF, and our aTTC-CBF; the results are summarized in Table~\ref{tab:3d_pursuit_results}.

\paragraph{Aggregate performance}
Without pursuers, the aTTC-CBF incurs roughly a third as many collisions as the HOCBF---relative to an already very low baseline---while achieving an approximately threefold higher waypoint-progress rate, only about $10\%$ below the no-CBF baseline. That the NN-based CBF achieves a lower collision rate \emph{despite} much faster forward progress in the non-adversarial case indicates that the uncertainty introduced into the forward-invariance condition by the surrogate's approximation error is small relative to the uncertainty incurred by the HOCBF encountering an infeasible constraint. With the introduction of slow pursuers, both CBFs drastically reduce the (now much higher) collision rate; here the aTTC-CBF again records less than half as many collisions as the HOCBF while maintaining a $1.8\times$ higher waypoint-progress rate. Notably, in the slow-pursuer case the aTTC-CBF actually exceeds the baseline progress rate, because the evaders travel at a higher average speed in their attempt to avoid the pursuers. These trends are exacerbated when the pursuers are faster than the evaders: the HOCBF progress rate drops to roughly $30\%$ of the no-CBF baseline, whereas the aTTC-CBF retains $58\%$ of the baseline rate while again recording half as many collisions as the HOCBF.

\paragraph{Mission-space spread}
By directly quantifying the estimated time-to-collision, the aTTC-CBF enables less conservative---and therefore more efficient---navigation by distinguishing between more and less dangerous configurations. We quantify this through the distance from each evader to the origin (the center of the mission space). The median values are reported in Table~\ref{tab:3d_pursuit_results}.
and the full PDFs are shown in Figure~\ref{fig:dist_pdf}. The spread of the HOCBF-equipped evaders grows quickly with pursuer speed, consistent with its falling waypoint-progress rate, whereas the aTTC-CBF spread rises only weakly. For the aTTC-CBF with no pursuers the median spread is nearly identical to the no-CBF baseline, rising by only $8\%$ and $37\%$ with the introduction of slow and fast pursuers, respectively. The HOCBF, while exhibiting a slightly lower-than-baseline median spread with no pursuers, shows a $20\%$ and $85\%$ increase above baseline under slow and fast pursuers. In summary, the aTTC-CBF outperforms the HOCBF across all cases and, critically, retains far more of its performance as the environment becomes increasingly adversarial.

\begin{figure}
    \centering
    \includegraphics[width=0.75\textwidth]{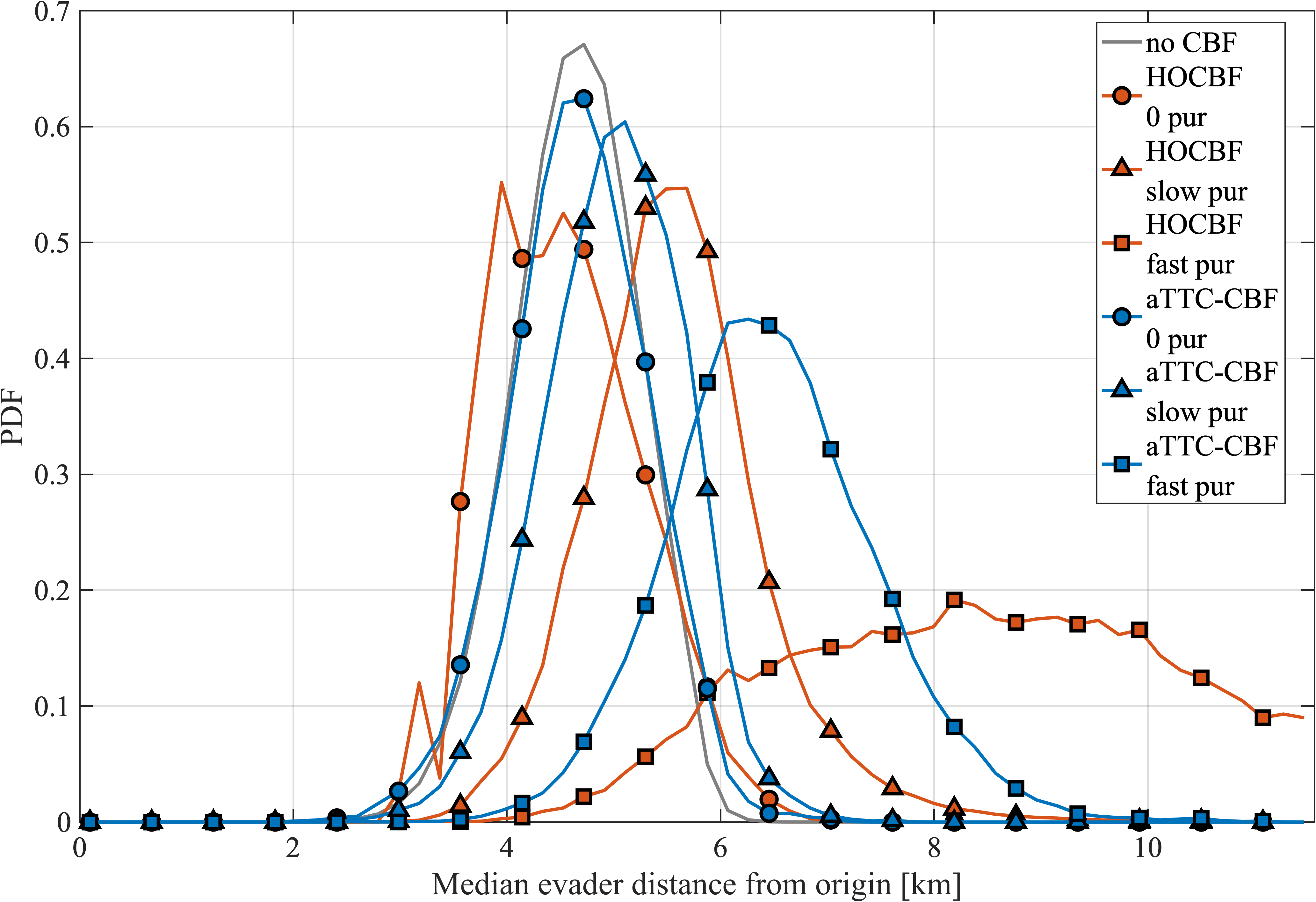}
    \caption{Probability density function of evader distance from origin summed over all evaders. aTTC-CBF (blue), HOCBF (red). Zero, slow, and fast pursuer scenarios show in circles, triangles, squares respectively. }
    \label{fig:dist_pdf}
\end{figure}

\paragraph{Near-miss escape probability}
To further analyze the collision-avoidance dynamics, we identify all
\emph{near-miss} encounters, defined as instances where the inter-agent distance drops below $4r_{col}$---twice the collision threshold. Let $t_{nm}$ be the time at which the inter-agent distance $d_{ij}$ first drops below $4r_{col}$ for a given pair; for a look-ahead time $T$ we then check whether a collision ($d_{ij} \le 2r_{col}$) occurs in the interval $[t_{nm}, t_{nm}+T]$. If so, we call the encounter \emph{unmitigated}; otherwise \emph{mitigated}. Figure~\ref{fig:random_sphere_col_frac} shows the fraction of unmitigated encounters---equivalently, $P(\text{collision}\mid\text{near-miss})$---as a function of $T$ for both pursuer cases. Most collisions occur by $T = 1$\,s. In the slow-pursuer case, the probability of collision given a close encounter is about five times lower for the aTTC-CBF than for the HOCBF. In the fast-pursuer case, the aTTC-CBF still escapes nearly $50\%$ of close encounters, whereas the HOCBF escapes only about $5\%$. Consistent with the aggregate results, the aTTC-CBF scales far more favorably with pursuer speed.

\paragraph{Evasive maneuver intensity}
To characterize \emph{how} the aTTC-CBF mitigates these encounters, we compute the average turning rate $|\dot{\theta}| \equiv \sqrt{\dot{\psi}^2 + \dot{\gamma}^2}$ over $[t_{nm}, t_{nm}+T]$ with $T = \min(1, T_{col})$, where $T_{col}$ is the time to collision. Figure~\ref{fig:random_sphere_turn} shows this metric across all near-miss encounters, with unmitigated collisions in red and mitigated encounters in blue; horizontal lines mark the group means, and the horizontal spacing of points is purely for visual clarity. As expected, for both controllers the mitigated encounters are associated with higher turning rates than the realized collisions---escape requires more drastic action. However, during these close encounters the aTTC-CBF sustains significantly higher turning rates than the HOCBF, and shows less separation between mitigated and realized cases. This indicates that the aTTC-CBF maneuvers more aggressively and operates closer to its maximum control limits---a consequence of its dynamics-aware formulation, which recognizes that collision risk can be reduced by changing relative velocity even in close proximity to another agent. The acceleration exhibits a similar trend, which we omit for brevity.

\begin{table*}[t]
\centering
\footnotesize
\renewcommand{\arraystretch}{1.15}
\caption{Statistics for the 3D independent multi-agent
pursuit-evasion experiment, by adversary regime.  Bold marks the better
method within each (metric, regime) pair.}
\label{tab:3d_pursuit_results}
\begin{tabular}{l@{\hspace{6pt}} c c c @{\hspace{10pt}} c c c @{\hspace{10pt}} c c c}
\hline
        & \multicolumn{3}{c}{\textbf{No pursuers}}
        & \multicolumn{3}{c}{\textbf{Slow pursuers}}
        & \multicolumn{3}{c}{\textbf{Fast pursuers}} \\
        & no CBF & HOCBF & aTTC-CBF
        & no CBF & HOCBF & aTTC-CBF
        & no CBF & HOCBF & aTTC-CBF \\
\hline
Collisions / 100\,s
        & 0.170  &        0.034  & \textbf{0.012}
        & 5.10   &        0.240  & \textbf{0.100}
        & 6.69   &        5.08   & \textbf{2.54} \\
Waypoints / evader / 100\,s
        & 2.19   &        0.677  & \textbf{1.99}
        & 2.19   &        1.31   & \textbf{2.41}
        & 2.19   &        0.656  & \textbf{1.27} \\
Median $\|\mathbf{r}\|$ [km]
        & 4.66   & \textbf{4.49} &        4.68
        & 4.66   &        5.58   & \textbf{5.04}
        & 4.66   &        8.64   & \textbf{6.40} \\
\hline
\end{tabular}
\end{table*}

\begin{figure}
    \centering
    \begin{subfigure}[t]{0.48\textwidth}
        \centering
        \includegraphics[width=\textwidth]{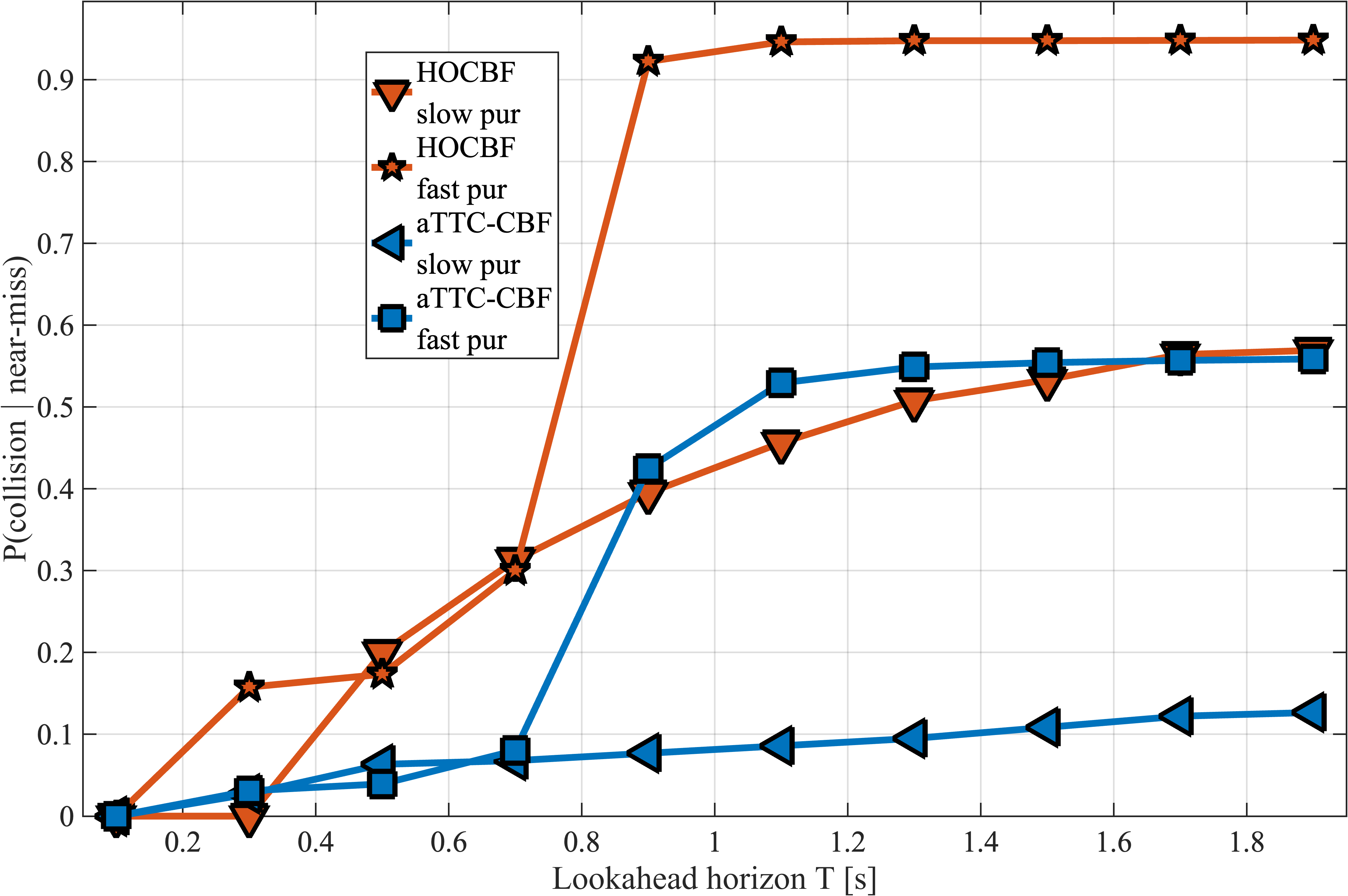}
        \caption{Fraction (probability) of collision following a close encounter as a function of look-ahead time $T$.}
        \label{fig:random_sphere_col_frac}
    \end{subfigure}
    \hfill
    \begin{subfigure}[t]{0.48\textwidth}
        \centering
        \includegraphics[width=\textwidth]{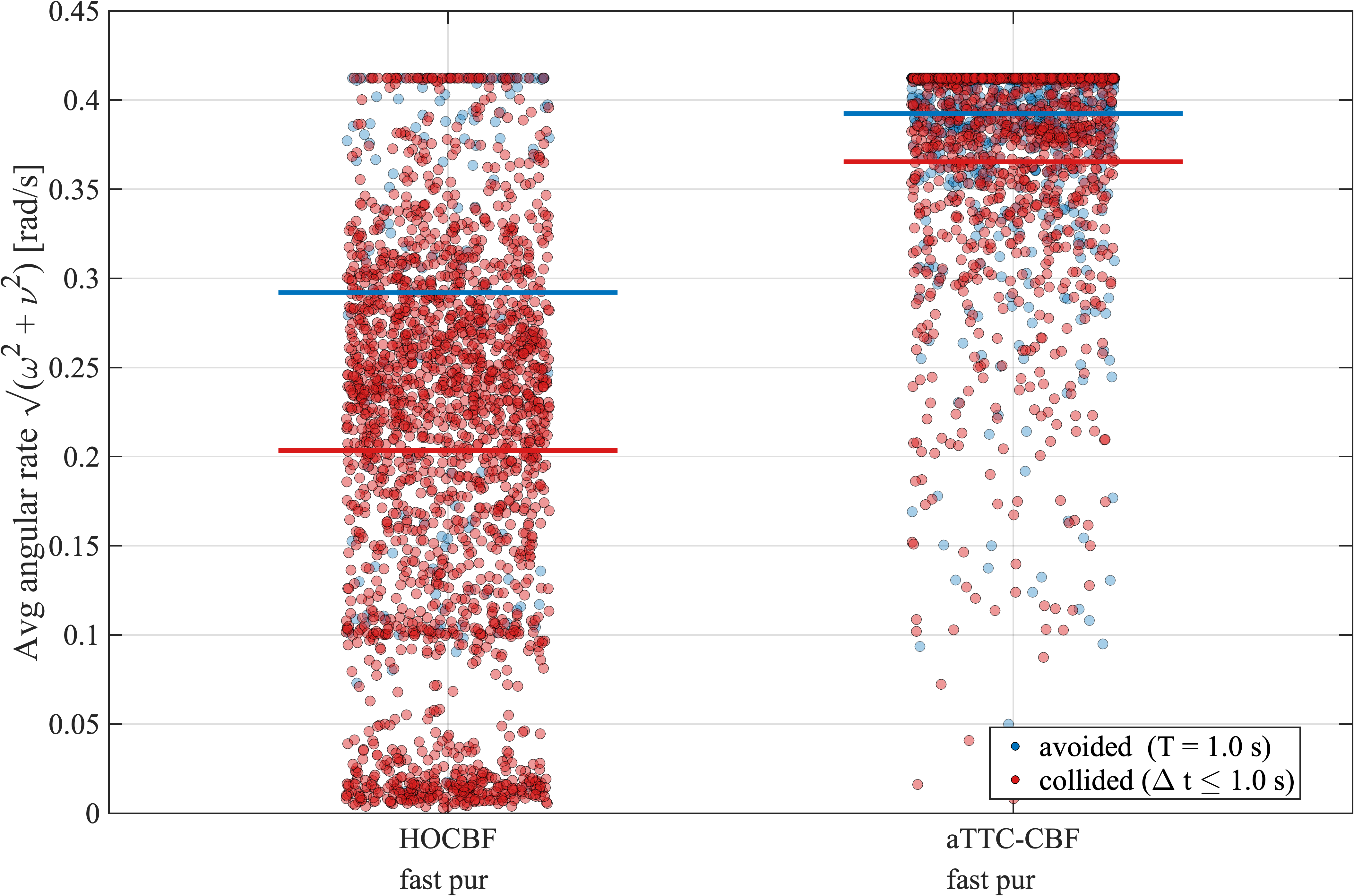}
        \caption{Average turning rate from near-miss encounter to collision or evasion.}
        \label{fig:random_sphere_turn}
    \end{subfigure}
    \caption{Near-miss dynamics for the independent multi-agent pursuit-evasion experiment.}
    \label{fig:random_sphere_nearmiss}
\end{figure}


\subsection{Formation Flight Pursuit Evasion}
Next we apply our aTTC-CBF to a collaborative mission environment. The eight evader agents fly loops through eight randomly spaced waypoints while in a chevron formation. The waypoints are  spread randomly throughout a domain of size $50\times40\times20$ km in the $x,y,z$ directions. The lead agent navigates towards the waypoint while the other evader agents navigate towards their (moving) assigned position in the chevron formation. We consider a chevron formation with an apex angle of $90^\circ$ and agent spacing -- along each arm -- of 2km. The three pursuer agents attempt to collide with the evaders using the same control and targeting algorithm described in the previous scenario. Although, we note that the CBF activation scales must be adjusted to allow for formation flight. We again consider both slower pursuer, $V^p_{max}/V^e_{max} = 0.9$, and faster pursuer, $V^p_{max}/V^e_{max} = 1.5$ and compare the aTTC-CBF and HOCBF -- we omit the no CBF and no pursuer baselines. The results are summarized in Table~\ref{tab:3d_formation_results}.

The overall collision and waypoint-progress rates are consistent with the independent-pursuit results. Under slow pursuers the two methods perform comparably: the HOCBF is marginally safer, avoiding collisions entirely, while the aTTC-CBF incurs a small but nonzero collision rate—comparable to the independent slow-pursuer case—and the two make nearly identical waypoint progress. The benefit of the aTTC-CBF emerges under fast pursuers, where it achieves roughly $9$\% fewer collisions and over $50\%$ greater waypoint progress. The more demanding test of the formation-flight scenario, however, is whether these collision and progress rates can be attained \emph{while holding the formation together}. This distinction matters because the formation mission is inherently less collision-dense than the previous example: where the spherically spaced waypoints forced a $\sim180^\circ$ turn at each target, producing frequent crossings of nominal trajectories, whereas the formation team tracks successive waypoints along largely parallel paths that seldom intersect.
\\
To quantify this coherence we measure the deviation—the Euclidean distance—of each evader from its nominal assigned slot behind the leader. The mean values are shown in Table~\ref{tab:3d_formation_results}. For the slow pursuer case both CBF methods have similar average deviations, yet the aTTC-CBF deviation remains largely unchanged with an increase in pursuer speed, while the HOCBF deviation nearly doubles.

Figure~\ref{fig:formation_deviation_pdf} shows the probability density of this deviation for the fast pursuer case accumulated over every timestep. We report two densities: one over all eight evaders (the \emph{full fleet}), and one that excludes, at each instant, the $N_p = 3$ most-deviated evaders (the \emph{core fleet}); the respective means are marked by dashed and dash-dotted vertical lines. Since the agents under active pursuit are precisely those driven furthest from formation, the core-fleet density approximates the coherence of the main formation once the $N_p$ pursued evaders are set aside. The aTTC-CBF fleet is markedly more coherent, with a mean deviation roughly half that of the HOCBF fleet. Moreover, the gap between the full- and core-fleet means is about $50\%$ smaller under the aTTC-CBF, indicating that even the actively-evading agents remain far closer to formation than their HOCBF counterparts. A qualitative illustration of these mechanisms is shown in Figure~\ref{fig:formation_trajs} which compares the trajectories of both CBF types -- evaders in blue, pursuers in red. Consistent with the previous results the aTTC-CBF formation is notably more compact. Note in particular, how the HOCBF equipped evaders targeted by the pursuers are chased far from the core of the formation -- a pathology not seen in the aTTC-CBF case.


\begin{figure}
    \centering
    \begin{subfigure}[b]{0.75\textwidth}
        \centering
        \includegraphics[width=\textwidth]{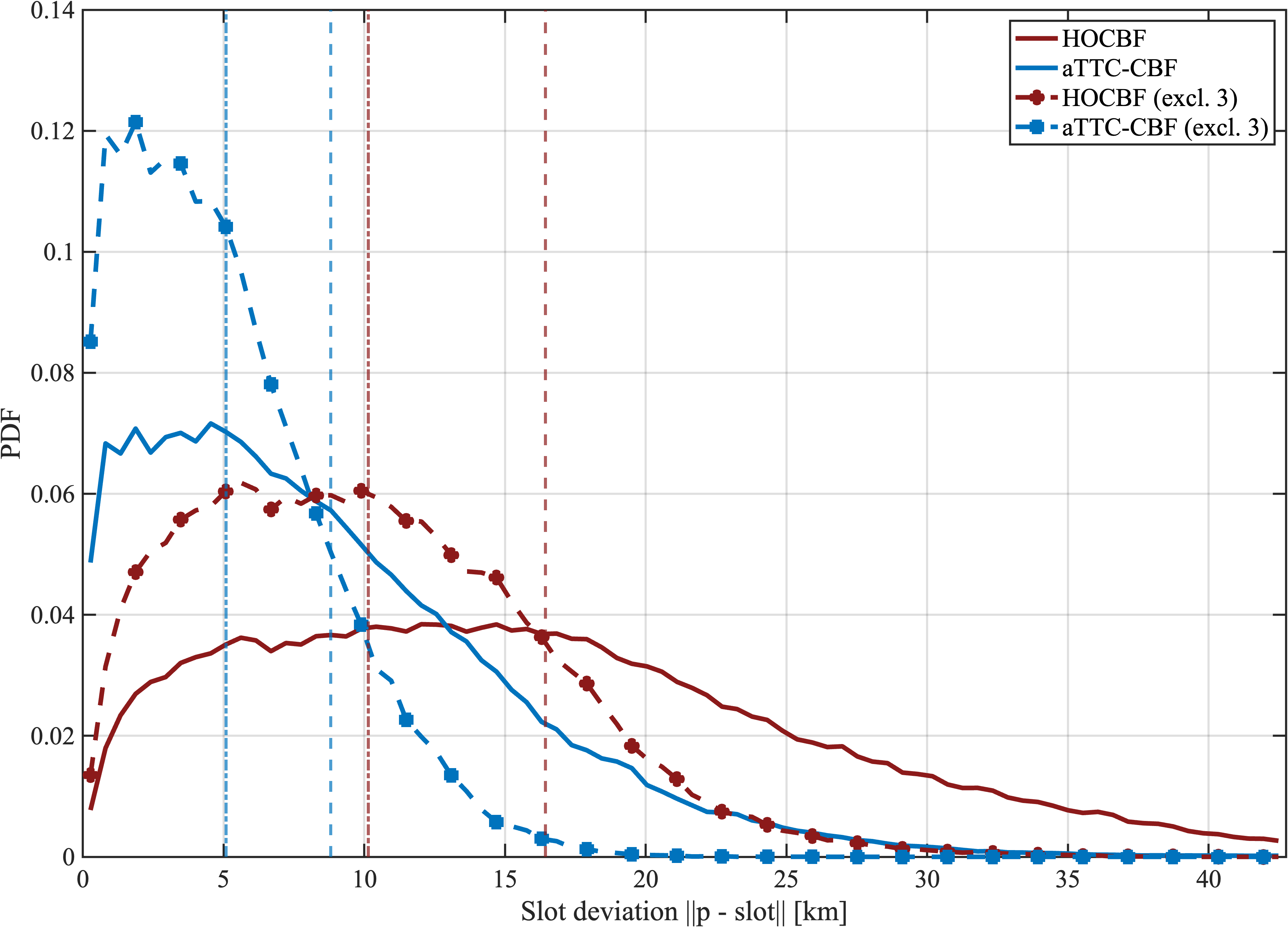}
        \caption{Probability density function of evader deviation from assigned formation location for aTTC-CBF (blue) and HOCBF (red). Full fleet (solid lines) and core fleet (dashed lines / circles).}
        \label{fig:formation_deviation_pdf}
    \end{subfigure}

    \vspace{1em}

    \begin{subfigure}[b]{0.95\textwidth}
        \centering
        \includegraphics[width=\textwidth]{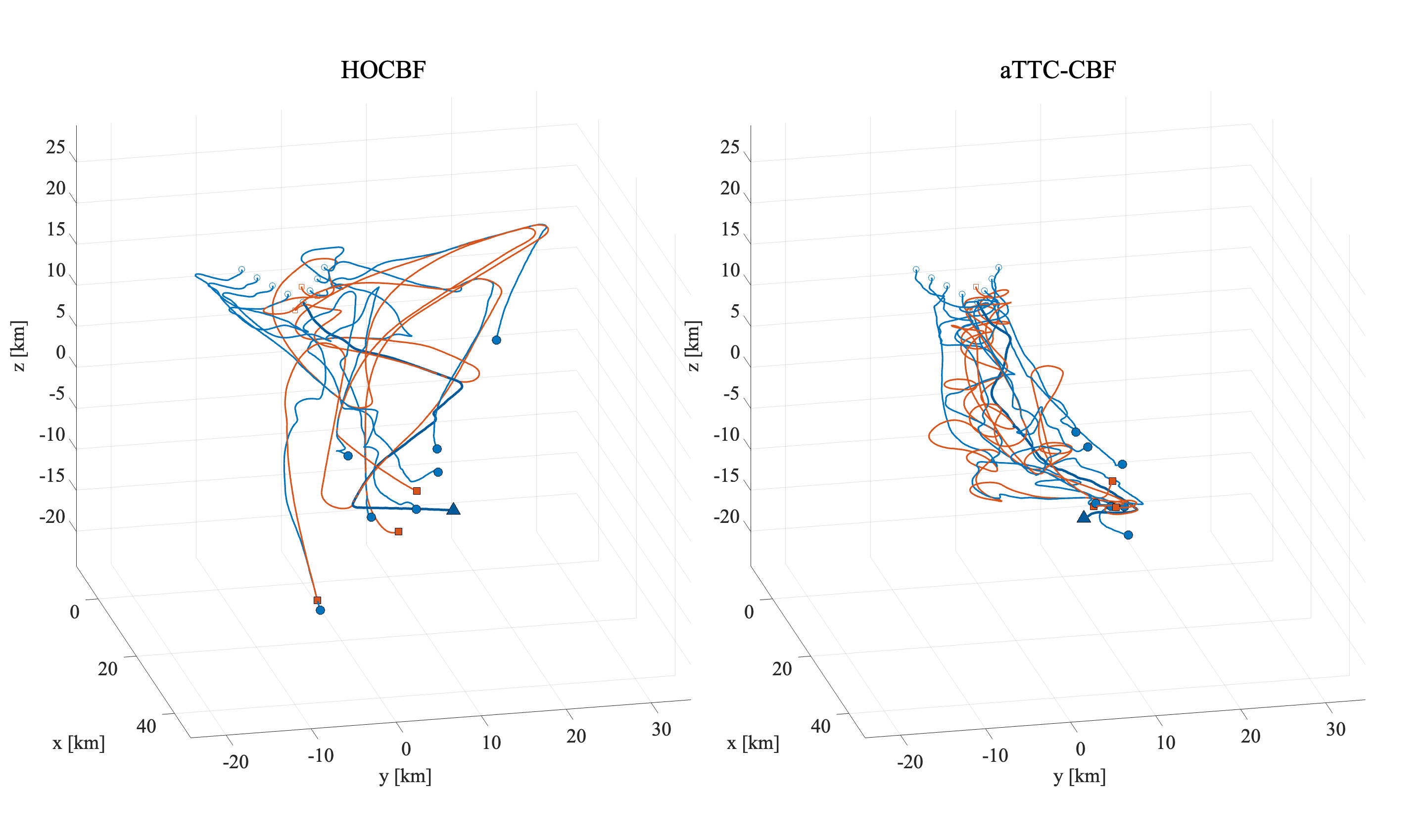}
        \caption{Trajectories for the first 200\,s of formation flight with HOCBF (left) and aTTC-CBF (right). Evaders are shown in blue, pursuers in red.}
        \label{fig:formation_trajs}
    \end{subfigure}
    \caption{Formation-flight pursuit-evasion under HOCBF vs.\ aTTC-CBF: formation coherence (a) and agent trajectories (b).}
    \label{fig:formation_combined}
\end{figure}

\begin{table*}[t]
\centering
\caption{Statistics for the 3D \emph{formation} multi-agent
pursuit-evasion experiment, by adversary regime.   Bold marks the
better method within each (metric, regime) pair. }
\label{tab:3d_formation_results}
\renewcommand{\arraystretch}{1.15}
\begin{tabular}{l@{\hspace{10pt}} c c @{\hspace{14pt}} c c}
\hline
                                & \multicolumn{2}{c}{\textbf{Slow pursuers}}
                                & \multicolumn{2}{c}{\textbf{Fast pursuers}} \\
                                & HOCBF          & aTTC-CBF
                                & HOCBF          & aTTC-CBF                 \\
\hline
Collisions / 100\,s
                                & \textbf{0.000} &        0.028
                                &        5.140   & \textbf{4.666}         \\
Waypoints / evader / 100\,s
                                &        1.54    & \textbf{1.58}
                                &        0.66    & \textbf{1.01}          \\
Mean slot deviation [km]
                                &        10.29   & \textbf{9.32}
                                &        16.43   & \textbf{8.81}          \\
Core slot deviation [km]
                                &        6.80    & \textbf{5.66}
                                &        10.14   & \textbf{5.08}          \\
\hline
\end{tabular}
\end{table*}


\section{Conclusions}\label{sec:conclusion}
We introduced the adversarial-time-to-collision (aTTC) as a robust worst-case risk metric for autonomous aerial vehicles. For a given ego agent, the aTTC is a dynamics aware metric which quantifies the estimated time-to-collision assuming maximally adversarial intent on behalf of any other surrounding agents. Quantifying collision risk directly in time rather than in position or velocity allows us to derive a metric which is intuitive in its interpretation and selective in its assessments. By integrating the system dynamics and assumed control authority limits we arrive at a metric which assigns low risk to agents on diverging trajectories -- even at very close proximity --  while agents on an intercept paths are appropriately marked as dangerous -- even at relatively large separations. Through a neural network based surrogate model we are able to incorporate the aTTC metric directly into a quadratic program based control barrier function collision avoidance algorithm -- coined the aTTC-CBF. When applied to both independent and formation-flight multi-agent pursuit evasion scenarios the selective nature of our aTTC-CBF enables agents to fly in closer proximity by maneuvering more aggressively -- and thus effectively -- between both friendly and adversarial agents. This leads to greater mission progress, lower collision rates, and more cohesive formation flight. The benefits of our aTTC-CBF were especially pronounced in increasingly adverserial settings. As pursuer speed increased the HOCBF performance degraded drastically, while our aTTC-CBF performance suffered only marginally.

Despite these advantages, some key limitations  must be acknowledged. First, while our aTTC-CBF framework demonstrated clear advantages over a distance based CBF in adversarial settings, in non-adversarial settings the improvement was more marginal. This implies that the current formulation may be too conservative in non-adversarial settings. Furthermore, the (current) definition of aTTC as the infimum over the domain of a functional can result in non-smooth gradients which require large amounts of training data to resolve. Both these challenges could be addressed by extending the aTTC from a single metric to a distribution over assumed pursuer intentions -- different measures of which could be used to quantify risk for different types of agents. Second, the current formulation assumes knowledge of the agent dynamics -- which may be unrealistic in some practical settings and would require data-driven updating of the assumed dynamics model. Finally, all data-driven models have some approximation error -- especially at the margins of the training data distribution. These may be mitigated through hybrid formulations which balance model based barriers when uncertainty is low and analytical barriers when uncertainty is high.
All these directions represent topics of ongoing work and we hope other authors find new and productive ways to incorporate our temporal barrier framework into their own research.

\section*{Data Availability}
The code used to to generate the data, train the neural networks, and generate all the results discussed herin can be found at \url{https://github.com/ben-barthel/Temporal_Barrier_Framework}.
\section*{Funding Sources}
DISTRIBUTION STATEMENT A. Approved for public release. Distribution is unlimited. This material is based upon work supported by the Department of the Air Force under Air Force Contract No. FA8702-15-D-0001 or FA8702-25-D-B002. Any opinions, findings, conclusions or recommendations expressed in this material are those of the author(s) and do not necessarily reflect the views of the Department of the Air Force. © 2026 Massachusetts Institute of Technology. Delivered to the U.S. Government with Unlimited Rights, as defined in DFARS Part 252.227-7013 or 7014 (Feb 2014). Notwithstanding any copyright notice, U.S. Government rights in this work are defined by DFARS 252.227-7013 or DFARS 252.227-7014 as detailed above. Use of this work other than as specifically authorized by the U.S. Government may violate any copyrights that exist in this work.
\section*{Appendix}
\appendix
\section{Surrogate Model}\label{app:nn-training}
The NN is trained using a weighted Huber loss
\begin{equation}    \mathcal{L}\left(\hat{\tau}^*,\tau^*\right)=\frac{1}{\tau^{*n}}\mathcal{L}_H\left(\hat{\tau}^*,\tau^*\right)
\end{equation}
where the $1/\tau^{*n}$ weight emphasizes low aTTC (high risk) configurations and the exponent $n$ tunes how heavily these are weighted. We found $n=3$ to give the best performance. We utilize this output value based weight as opposed to probability-based weighted losses \cite{rudy_output-weighted_2021}, as these overemphasize the rare, but unimportant, large values of $\tau^*$. The Huber loss itself is defined as
\begin{equation*}
\mathcal{L}_H\left(\hat{\tau}^*,\tau^*\right) =
    \begin{cases}
\frac{1}{2}(\tau^* - \hat{\tau}^*)^2, & \text{if } |\tau^* - \hat{\tau}^*| \leq \delta \\
\delta \left( |\tau^* - \hat{\tau}^*| - \frac{1}{2}\delta \right), & \text{otherwise}
\end{cases}
\end{equation*}
where $\delta$ is a user defined cutoff we set to $1$s. The Huber loss \cite{huber_robust_1964} applies a mean-squared-error (MSE) loss to small outputs and a linear loss to larger values. This ensures smooth behavior near zero without overly weighting errors for large values of $\tau^*$.  
\begingroup
\renewcommand{\url}[1]{}

\bibliography{refs/references_auspice}
\endgroup

\end{document}